%% file: paper.tex
\documentclass[11pt]{article}
\title{Incremental DFS algorithms: a theoretical and experimental study}
\author{Surender Baswana\thanks{Department of Computer Science and Engineering, Indian Institute of Technology Kanpur, India (www.cse.iitk.ac.in),
email: \{sbaswana,shahbazk\}@cse.iitk.ac.in, ayushgoel529@gmail.com. 
}
\thanks{This research was partially supported 
by {\em UGC-ISF} (the University Grants Commission of India \& Israel Science Foundation) and
{\em IMPECS} (the Indo-German Max Planck Center for Computer Science) .}
\and Ayush Goel \footnotemark[1]
\and Shahbaz Khan \footnotemark[1] \thanks{This research was partially supported 
by Google India under the Google India PhD Fellowship Award.}
}

\usepackage{amssymb,amsthm}
\usepackage{amsmath}
\usepackage{graphicx,array}
\usepackage{hyperref}

\newtheorem{theorem}{Theorem}[section]

\newtheorem{lemma}{Lemma}[section]

\newtheorem{corollary}{Corollary}[section]

\begin{document}

\date{}
\maketitle

\input{abstract}
%\newpage
\input{intro}
\input{prelim}
\input{basic_exp_random_graph}
\input{broomStick_exp}

\input{probability_analysis}

\input{broomStick_app}
\input{exp_real_data}
\input{conclusion}
%
%\balance

%\bibliographystyle{plain}
\bibliography{paper.bib}

\appendix
\input{appendix}
%\end{appendix}
%\balance

\end{document}

%% file: abstract.tex
\begin{abstract}
Depth First Search (DFS) tree is a fundamental data structure used for solving various graph problems.
For a given graph $G=(V,E)$ on $n$ vertices and $m$ edges, a DFS tree can be built in $O(m+n)$ time.
In the last 20 years, a few algorithms have been designed for maintaining a 
DFS tree efficiently under insertion of edges. 
For undirected graphs, there are two prominent algorithms, namely, ADFS1 and ADFS2~[ICALP14] %\cite{BaswanaK14} 
that achieve total update time of $O(n^{3/2}\sqrt{m})$ and  $O(n^2)$ respectively. 
For directed acyclic graphs, the only non-trivial algorithm, namely, FDFS~[IPL97] %\cite{FranciosaGN97} 
requires total 
$O(mn)$ update time. However, even after 20 years of this result, there does not exist any 
nontrivial incremental algorithm for a DFS tree in directed graphs with $o(m^2)$ worst case bound. 

In this paper, we carry out extensive experimental and theoretical 
evaluation of the existing algorithms for maintaining incremental DFS in random graphs 
and real world graphs and derive the following interesting results.
\begin{enumerate}
\item
%%%ADFS1 and ADFS2 perform equally well on random graphs. 
%%%This is quite surprising because the worst case total update time of ADFS1 is
%%%$\sqrt{m/n}$ times that of ADFS2.
For insertion of uniformly random sequence of $n \choose 2$ edges, each of ADFS1, ADFS2 and FDFS 
perform equally well and are found to take $\Theta(n^2)$ time experimentally. 
This is quite surprising because the worst case bounds of ADFS1 and FDFS
are greater than $\Theta(n^2)$ by a factor of $\sqrt{m/n}$ and $m/n$ respectively,
which we prove to be tight.
%Furthermore, we present worst case examples to prove that the theoretical bounds of ADFS1 and 
%FDFS are indeed tight.
We complement this experimental result by probabilistic analysis of these algorithms 
proving $\tilde{O}(n^2)$ \footnote{$\tilde{O}()$ hides the poly-logarithmic factors.} bound on the time complexity.
%%%%%% On 10 July
%%that matches
%%the $\Theta(n^2)$ bound up to polylogarithmic factors. %behavior quite accurately. 
%%%%%%
%For this purpose, we derive results about the structure of a DFS tree in a random graph. 
%
%We then present theoretical proof for their exceptional behavior in random graphs 
%giving high probability bounds that resemble their empirical behavior. 
For this purpose, we derive results about the structure of a DFS tree in a random graph. 
These results are of independent interest in the  domain of random graphs.
% and could very well 
%pave way for other efficient algorithms for random graphs that are based on DFS tree. 

%This is quite surprising because the worst case total update time of ADFS1 is
%$\sqrt{m/n}$ times that of ADFS2.
%in their total update time. 
%We even present worst case examples to show that this gap is indeed tight. 
%\item
%Each of ADFS1, ADFS2 and FDFS perform exceptionally well for random graphs. 
%They achieve close to $O(m)$ total update time. 
\item
The insight that we developed about DFS tree in random graphs leads us to design
an extremely simple algorithm for incremental DFS that works for both undirected and directed graphs.
Moreover, this algorithm theoretically matches and experimentally outperforms the performance of 
the state-of-the-art algorithm in dense random graphs.
%Moreover, this algorithm matches or outperforms the performance of all the previous existing algorithms
%theoretically as well as experimentally in dense random graphs.
%Thus, we propose very simple algorithm for maintaining incremental DFS 
%that performs equally well (as ADFS/FDFS) for both undirected and directed graphs.
Furthermore, it can also be used as a single-pass semi-streaming algorithm for 
computing incremental DFS
%, topological ordering 
and strong connectivity for random graphs using $O(n\log n)$ space.

%\item Our proposed algorithm can also be used as a single-pass semi-streaming algorithm for 
% incremental DFS tree of a random graph using $O(n\log n)$ space.
%This can be directly used to compute topological ordering and strong connectivity in the same setting.
%%
%% gives a single pass semi-streaming algorithm for evaluating topological ordering and strong connectivity
%%queries on random graphs.
%It has been shown that no single-pass semi-streaming algorithm can evaluate these queries deterministically on general graphs.
%This distinguishes the hardness of these problems for general and random graphs,
%which was not known earlier for any graph problem to the best of our knowledge.% no such result .
%
%, strong bi-connectivity and strong 2-edge connectivity queries for random graphs.
%These problems are known to require $\Omega(m)$ computing space for single-pass streaming algorithms in general graphs.
%%%We are also able to provide a theoretical explanation for the exceptional behavior 
%%%of these algorithms. For this purpose, we derive results about the structure of a 
%%%DFS tree in a random graph. These results are of independent interest in the 
%%%domain of random graphs and could very well pave way for other efficient 
%%%algorithms for random graphs that are based on DFS tree.  
\item
Even for real world graphs, which are usually sparse, both ADFS1 and FDFS turn out to be much better than
their theoretical bounds. Here again, we present two simple algorithms for incremental DFS for 
directed and undirected graphs respectively, which perform very well on real graphs.
In fact our proposed algorithm for directed graphs almost always matches the performance of FDFS. 

%Finally, we present two simple algorithms for maintaining incremental DFS in real graphs.
%These algorithms are inspired by ADFS1 and FDFS which almost match their performance in
%most cases. 
%
%Finally, we present two simple algorithms for maintaining incremental DFS in real graphs.
%These algorithms are inspired by ADFS1 and FDFS which almost match their performance in
%most cases. 
%%%As a result, we are able to design much simpler incremental algorithms for 
%%%a DFS tree for undirected graphs and even directed graphs. 
%%%In fact, these algorithms match the performance of ADFS and FDFS in random graphs.
%%%Moreover, they also work extremely well for real world graphs.
%%%%These algorithms work extremely well for random graphs.
%%%%These algorithms also take close to $O(m)$ total time for $m=\Omega(n \log n)$ edge insertions. 
%%%Finally, we also present a single-pass semi-streaming algorithm for 
%%%computing DFS tree of a random graph that takes $O(n\log n)$ space.
\end{enumerate}

\end{abstract}

%% file: intro.tex
\section{Introduction}
%\vspace{-.75em}
Depth First Search (DFS) is a well known graph traversal technique. Right from the seminal work of Tarjan~\cite{Tarjan72}, DFS traversal 
has played the central role in the design of efficient algorithms for many fundamental graph problems, namely, biconnected components, 
strongly connected components, topological sorting~\cite{Tarjan72},
bipartite matching~\cite{HopcroftK73}, dominators~\cite{Tarjan74} and planarity testing~\cite{HopcroftT74}.
%Interestingly, the role of DFS traversal is not confined to merely the design of efficient algorithms. 
%For example, consider the classical result of Erd{\H o}s and R{\' e}nyi~\cite{ErdosR60} for the 
%phase transition phenomena in random graphs. There exist many proofs of this result which are intricate and based on highly sophisticated 
%probability tools. However, recently, Krivelevich and Sudakov~\cite{KrivelevichS13}  designed a truly simple, short, and elegant proof 
%for this result based on the insights from a DFS traversal in a graph.

A DFS traversal produces a rooted spanning tree (or forest), called DFS tree (forest). 
Let $G=(V,E)$ be a graph on $n=|V|$ vertices and $m=|E|$ edges. 
It takes $O(m+n)$ time to perform a DFS traversal and generate its DFS tree (forest). 
Given any ordered rooted spanning tree, the non-tree edges of the graph can be %partitioned 
classified into four categories as follows. %, namely, back edges, forward edges, cross edges and anti-cross edges. 
An edge directed from a vertex to its ancestor in the tree is called a {\it back edge}. 
Similarly, an edge directed from a vertex to its descendant in the tree is called a {\it forward edge}.
Further, an edge directed from right to left in the tree is called a {\it cross edge}.
The remaining edges directed from left to right in the tree are called {\it anti-cross edges}.
A necessary and sufficient condition for such a tree to be a DFS tree is the absence of anti-cross edges. 
In case of undirected graphs, this condition reduces to absence of all cross edges.

%Given any rooted spanning tree of graph $G$, all non-tree edges of the graph can be %partitioned 
%classified into four categories, namely, back edges, forward edges, cross edges and backward cross edges
%as follows. A non-tree edge is called a {\it back edge} if it is directed from a vertex to its ancestor in the tree. 
%Similarly, a non-tree edge is called a {\it forward edge} if it is directed from a vertex to its descendant in the tree.
%The remaining edges can be classified as {\it cross edges} if it is directed from a vertex to some vertex visited before it by the DFS traversal.
%Otherwise, it is called a {\it backward cross edge}. A necessary and sufficient condition for any rooted 
%spanning tree to be a DFS tree is the absence of backward cross edges. In case of undirected graphs, this condition
%reduces to absence of all cross edges.

Most of the graph applications in the real world deal with graphs that keep 
changing with time. These changes can be in the form of insertion or deletion of edges. 
An algorithmic graph problem is modeled in
the dynamic environment as follows. There is an online 
sequence of insertion and deletion of edges and the aim is to maintain the
solution of the given problem after every edge update. To achieve this aim,
we need to maintain some clever data structure for the problem such that the time taken to   
update the solution after an edge update is much smaller than 
that of the best static algorithm. A dynamic algorithm is called an 
incremental algorithm if it supports only insertion of edges.
In spite of the fundamental nature of the DFS tree,  very few incremental algorithms have been designed for  a 
DFS tree. A short summary of the current-state-of-the-art of incremental algorithms for DFS tree is as follows.
An obvious incremental algorithm is to recompute the whole DFS tree in $O(m+n)$ time from scratch after every edge insertion. Let us
call it SDFS henceforth. It was shown by Kapidakis~\cite{Kapidakis90} that a DFS tree can be computed in $O(n\log n)$ time
on random graph~\cite{ErdosR60,Bollobas84} if we terminate the traversal as soon as all vertices are visited. 
Let us call this variant as SDFS-Int.
%shown to have much better performance for a random graph by 
%. Only difference from SDFS is that the algorithm terminates as soon as all the vertices of the graph 
%are marked visited. 
Notice that both these algorithms recompute the DFS tree from scratch after every edge insertion. Let us now
move onto the algorithms that avoid this recomputation from scratch.

The first incremental DFS algorithm, namely FDFS, 
was given by Franciosa et al.~\cite{FranciosaGN97} for directed acyclic graphs,
requiring total update time of $O(mn)$. 
For undirected graphs, Baswana and Khan~\cite{BaswanaK14} presented two algorithms, 
namely ADFS1 and ADFS2,  that achieve total update time of $O(n^{3/2}\sqrt{m})$ and 
$O(n^2)$ respectively. However, the worst case update time to process an edge for these 
algorithms is still $O(m)$. 
Recently, an incremental algorithm~\cite{BaswanaCCK16}, namely WDFS, giving a worst case guarantee of 
$O(n\log^3 n)$ on the update time was designed. 
%Recently, Baswana et al. designed an incremental algorithm~\cite{BaswanaCCK16}, namely WDFS, giving a 
%worst case guarantee of $O(n\log^3 n)$ on the update time. 
However, till date there is no non-trivial incremental algorithm for DFS tree in 
general directed graphs. Refer to Table~\ref{tab:exist_algo} for a comparison of these results.

Despite having several algorithms for incremental DFS, not much is known about their empirical performance. %analysis. 
%The average-case complexity denotes the average performance of an algorithm on random graphs.
For various algorithms~\cite{Meyer01a,AlimontiLM96,BastMST06}, 
the average-case time complexity (average performance on random graphs) have been proven to 
be much less than their worst case complexity. 
%For various algorithms, the average-case time analysis for random graphs have been extensively studied~\cite{Meyer01a,AlbertsH98,AlimontiLM96} proving their average case complexity to be 
%much less than their worst case complexity. 
A classical example is the algorithm by Micali and Vazirani~\cite{MicaliV80} for maximum matching. 
Its average case complexity have been proved to be only $O(m\log n)$~\cite{Motwani94,BastMST06},
despite having a worst case complexity of $O(m\sqrt{n})$. 
%despite requiring $O(m\sqrt{n})$ time in the worst case. 
%A classical example is of the algorithm by Micali and Vazirani~\cite{MicaliV80} for maximum matching, 
%requiring $O(m\sqrt{n})$ time in worst case, have later been proved~\cite{Motwani94,BastMST06} 
%to have an average case complexity of only $O(m\log n)$. 
An equally important aspect is the empirical performance of an algorithm on real world graphs.
After all, the ideal goal is to design an algorithm having a theoretical guarantee of efficiency
in the worst case as well as superior performance on real graphs. 
%The experimental analysis of an algorithm bridges the gap between theory and practice. 
Often such an experimental analysis also leads to the design of simpler algorithms 
that are extremely efficient in real world applications. 
Thus, such an analysis bridges the gap between theory and practice. 
The experimental analysis of different algorithms for several dynamic graph problems have been performed
including connectivity~\cite{AlbertsCI97,IyerKRT01}, minimum spanning trees~\cite{RibeiroT07,CattaneoFPI10}, 
shortest paths~\cite{DemetrescuI06,GorkeMSSW13,SchultesS07}, etc. 

%\vspace{0em}
\begin{table}[!ht]
\centering
\begin{tabular}{|m{1.8cm}|l|m{3.5cm}| m{3.5cm} |}
  	\hline
    \textbf{Algorithm} & \textbf{Graph} & \textbf{Time per update} & \textbf{Total time}\\
    \hline
	SDFS~\cite{Tarjan72} & Any & $O(m)$ & $O(m^2)$ \\ \hline
%	SDFS-Int~\cite{Kapidakis90} & Random & $O(n\log n)$ ~~~~expected & $O(mn\log n)$ expected \\ \hline
	SDFS-Int~\cite{Kapidakis90} & Random & $O(n\log n)$  expected & $O(mn\log n)$ expected \\ \hline
%	FDFS~\cite{FranciosaGN97} & DAG & $O(n)$ \qquad amortized & $O(mn)$ \\ \hline
	FDFS~\cite{FranciosaGN97} & DAG & $O(n)$  amortized & $O(mn)$ \\ \hline
	ADFS1~\cite{BaswanaK14} & Undirected & $O(n^{3/2}/\sqrt{m})$ amortized & $O(n^{3/2}\sqrt{m})$ \\ \hline
%	ADFS1~\cite{BaswanaK14} & Undirected & $O(n^{3/2}/\sqrt{m})$ amortized & $O(n^{3/2}m^{1/2})$ \\ \hline
	ADFS2~\cite{BaswanaK14} & Undirected & $O(n^2/m)$ amortized & $O(n^2)$ \\ \hline
	WDFS~\cite{BaswanaCCK16} & Undirected & $O(n\log^3 n)$ & $O(mn\log^3 n)$ \\	\hline
  \end{tabular}
\caption{Comparison of different algorithms for maintaining incremental DFS of a graph.}
\label{tab:exist_algo}
\end{table}

%Our study focuses on dynamic DFS algorithms only in the incremental setting as most dynamic graphs 
%in real world are dominated by insertion updates~\cite{Konect13,snapDATA,bigDND14}. 
Our study focuses on only incremental DFS algorithms as most dynamic graphs in real world are 
dominated by insertion updates~\cite{Konect13,snapDATA,bigDND14}. 
Moreover, in every other dynamic setting, only a single dynamic DFS algorithm is known~\cite{BaswanaCCK16,BaswanaC15}, 
making a comparative study impractical.
%Moreover, for other dynamic settings multiple non-trivial algorithms~\cite{BaswanaCCK16,BaswanaC15} 
%are not known for maintaining a DFS tree, making a comparative study impractical.

%It is also unknown how good these
%algorithm are on real world graph undergoing insertion of edges. 
%empirical performance in 
%random or real graphs. Experimental analysis of existing algorithms bridges the gap between theory and practice,
%reporting the algorithms for practical use along with efficient heuristics for tuning etc. Algorithms for various dynamic graph
%problems have been experimentally evaluated including connectivity~\cite{AlbertsCI97,IyerKRT01}, minimum spanning trees~\cite{RibeiroT07,CattaneoFPI10},
%shortest paths~\cite{DemetrescuI06,GorkeMSSW13,SchultesS07}
%%, dominators~\cite{GeorgiadisILS12}12}, bi-connectivity~\cite{FirmaniGILS16}, 2-edge connectivity~\cite{FirmaniGILS16}, clustering~\cite{BauerW09} 
%etc. Motwani~\cite{Motwani94} studied the behavior of the classical algorithm for maximum matching by Micali and Vazirani~\cite{MicaliV80} in random graphs. 
%Though the algorithm takes $O(m\sqrt{n})$ time in general graphs, Motwani~\cite{Motwani94} proved that it takes just $O(m)$ time in random graphs.

%\vspace{-1em}
\subsection{Our result}
%\vspace{-.5em}
In this paper, we contribute to both experimental analysis and average-case analysis of the algorithms for incremental DFS.
Our analysis reveals the following interesting results.
%In this paper, we carry out extensive experimentation and theoretical analysis of the incremental algorithms for DFS tree 
%in random graphs and derive the following interesting results.
%
%For insertion of uniformly random sequence of $n \choose 2$ edges, each of ADFS1, ADFS2 and FDFS 
%perform equally well and are found take $\Theta(n^2)$ time experimentally. 
%This is quite surprising because the worst case bounds of ADFS1 and FDFS
%are greater than $\Theta(n^2)$ by a factor of $\sqrt{m/n}$ and $m/n$ respectively,
%which are also shown to be tight.
%We complement this experimental result by probabilistic analysis of these algorithms that matches
%the $\Theta(n^2)$ bound up to polylogarithmic factors. %behavior quite accurately. 
%For this purpose, we derive results about the structure of a DFS tree in a random graph. 
%These results are also of independent interest in the  domain of random graphs.
\begin{enumerate}
\item \textbf{Experimental performance of the existing algorithms}\\
%\subsubsection*{Experimental performance of the existing algorithms}
%\subsubsection*{Experimental analysis of the existing algorithms}
We first evaluated the performance of the existing algorithms on the insertion of a 
uniformly random sequence of $n \choose 2$ edges.
The most surprising revelation of this evaluation was the similar performance of ADFS1 and ADFS2, 
despite the difference in their worst case bounds (see Table~\ref{tab:exist_algo}). 
Further, even FDFS performed better on random graphs taking just $\Theta(n^2)$ time.
This is quite surprising because the worst case bounds of ADFS1 and FDFS
are greater than $\Theta(n^2)$ by a factor of $\sqrt{m/n}$ and $m/n$ respectively. % (see Table~\ref{tab:exist_algo}).
We then show the tightness of their %worst case 
analysis~\cite{BaswanaK14,FranciosaGN97} %of ADFS1 and FDFS %is indeed tight, 
by constructing worst case examples for them %these algorithms 
(see Appendix~\ref{appn:wc_ADFS} and~\ref{appn:wc_FDFS}).
%Their superior performance in random graphs motivated us to explore the structure of a DFS tree in a random graph.
Their superior performance on random graphs motivated us to explore 
the structure of a DFS tree in a random graph.
%
%ADFS1 and ADFS2 perform equally well on random graphs. First we tried to find whether the analysis
%of ADFS1 is loose. But we could construct a worst case sequence of $m$ edge insertions such that ADFS1 will indeed
%take $\Omega(n^{3/2}\sqrt{m})$ time to process this sequence. In the view of this $\sqrt{m/n}$ gap,
%the experimental result showing exactly the same performance of ADFS1 and ADFS2 was indeed quite surprising.
%Moreover, it was found that the performance of each of ADFS1, ADFS2, and FDFS was significantly
%better on random graphs than their worst case bound. Essentially, each of them achieve close to $O(m)$ total update time. These results motivated us to explore the structure of a DFS tree in a random graph.
\item \textbf{Structure of DFS tree in random graphs}\\
%\subsubsection*{Structure of DFS tree in random graphs}
A DFS tree of a random graph can be seen as a broomstick:  a possibly long path without any branching (stick) 
followed by a bushy structure (bristles).
%We show that as the graph gets denser, the size of stick would increase significantly.
As the graph becomes denser, we show that the length of the stick would increase significantly
%this stick will have no branching upto a significant depth 
and establish the following result.
\begin{theorem}
For a random graph $G(n,m)$ with $m=2^i n\log n$, its DFS tree will have a stick of length at least $n-n/2^i$ 
with probability $1-O(1/n)$.
\label{thm:bristle-sizeM}
\end{theorem}
%
%There already exists some work on DFS tree in a random graph. Sibeyn showed that a DFS tree of random graph 
%$G(n,m)$ looks like a broomstick for $m>n\log n$: a long path with very few branching (stick) followed by a bushy structure (bristles). 
%We define the broomstick in a different way: where the {\em stick} is not allowed to have any branching.
%This restriction proved to be very significant in deriving our results.
%
%
%For a random graph on $m>n\log n$ edges, the DFS tree will have the following broomstick structure with probability
%1 asymptotically: A path of length $>n-n^2 \log n/m$ without any branching followed by a bushy structure. 
%In precise words, we state and prove the following result.
%
%\begin{theorem}
%For any given $n,m$, and any increasing function $c(n)$, let $n_0> \log n$ be such that $ n_0 (\log n_0 +c) =\frac{n_0^2}{n^2}m $. 
%A DFS tree in $G(n,m)$ will appear as a broomstick with stick length at least $n-n_0$ with probability 1 as $n$ tends to $\infty$. 
%\label{theorem:broomstick-result}
%\end{theorem}
The length of stick evaluated from our experiments matches perfectly with the value given by Theorem~\ref{thm:bristle-sizeM}.
It follows from the broomstick structure that the insertion of only the edges with both endpoints in the bristles
can change the DFS tree. % whose both endpoints lie in the bristles. 
%It follows from the broomstick structure that the insertion of only those edges can
%change the DFS tree whose both endpoints lie in the bristles. 
As follows from Theorem~\ref{thm:bristle-sizeM}, the size of bristles decreases %becomes smaller 
as the graph becomes denser. 
With this insight at the core, we are able to establish $\tilde{O}(n^2)$ bound on ADFS1 and FDFS for a uniformly random 
sequence of $n \choose 2$ edge insertions.

\noindent
\textbf{Remark:} It was Sibeyn~\cite{Sibeyn01} who first suggested viewing a DFS tree as a broomstick while studying the 
height of a DFS tree in random graph. However, his definition of {\em stick} allowed a few branches on the stick as well. 
Note that our added restriction (absence of branches on the stick) is extremely crucial 
in deriving our results as is evident from the discussion above.

\item \textbf{New algorithms for random and real world graphs}\\
%\subsubsection*{New algorithms for random and real world graphs}
We use the insight about the broomstick structure and Theorem~\ref{thm:bristle-sizeM} to 
design a much simpler incremental DFS algorithm (referred as SDFS2) that works for both undirected graphs and directed graphs. 
%Despite being very simple, it is shown to match the performance of ADFS and FDFS for dense random graphs
%both theoretically as well as experimentally. 
Despite being very simple, it is shown to theoretically match (upto $\tilde{O}(1)$ factors) 
and experimentally outperform the performance of ADFS and FDFS for dense random graphs.
%both theoretically as well as experimentally. 

For real graphs both ADFS and FDFS were found to perform much better than other algorithms including SDFS2.
With the insights from ADFS/FDFS, we design two simple algorithms 
for undirected and directed graphs respectively (both referred as SDFS3), which perform much better than SDFS2. 
In fact, for directed graphs SDFS3 almost matches the performance of FDFS for most real graphs considered, 
despite being much simpler to implement as compared to FDFS. 
%In fact, SDFS3 for directed graphs almost matches the performance of FDFS for most of the considered real datasets, 
%despite being much simpler to implement as compared to FDFS. 

%Again, our proposed algorithms almost match the performance of ADFS/FDFS 
%for most of the considered datasets, despite being much simpler to implement as compared to ADFS/FDFS. 
%[REAL GRAPHS]

\item \textbf{Semi-Streaming Algorithms}\\
%\subsubsection*{Semi-Streaming Algorithms}
Interestingly, both SDFS2 and SDFS3 can also be used as single-pass semi-streaming algorithms 
for computing a DFS tree of a random graph using $O(n\log n)$ space.
This immediately also gives 
a single-pass semi-streaming algorithm using the same bounds for 
%single-pass semi-streaming algorithms using the same bounds for 
%maintaining topological ordering and 
answering strong connectivity 
%, strong bi-connectivity points and strong 2-edge connectivity 
queries incrementally.
Strong connectivity is %These problems are 
shown~\cite{BorradaileMM14,Janthong14} to require a working memory of $\Omega(\epsilon m)$ to answer these queries 
with probability greater than $(1+\epsilon)/2$ in general graphs.
Hence, our algorithms not only give a solution for the problem
%these problems 
in semi-streaming setting 
but also establish the difference in hardness of 
the problem 
%these problems 
in semi-streaming model for general and random graphs.
\end{enumerate}

%\vspace{-.25em}
\subsection{Organization of the article}
We now present the outline of our paper.
In Section~\ref{sec:prelim}, we describe the various notations used throughout the paper in addition to the 
experimental settings and the datasets used as input. Section~\ref{sec:exist_algos} gives a brief overview of the 
existing algorithms for maintaining incremental DFS.
%In Section~\ref{sec:prelim}, we describe the experimental setting, datasets used as input and various notations used throughout the paper.
%Section~\ref{sec:exist_algos} gives a brief overview of the important existing algorithms for incremental maintenance of DFS tree.
The experimental evaluation of these algorithms on random undirected graphs is presented in Section~\ref{sec:basic_exp}.
In the light of inferences drawn from this evaluation, the experiments to understand the structure of the DFS tree for random graphs is 
presented in Section~\ref{sec:broomStick_exp}. Then, we theoretically establish the properties
of this structure and provide a tighter analysis of the aforementioned algorithms for random graphs in Section~\ref{sec:BroomStickImp}. 
The new algorithm for incremental DFS inspired by the broomstick structure of the DFS tree 
%properties of DFS tree in random graphs 
is presented and evaluated 
in Section~\ref{sec:new_algos}.
%We also compare these proposed algorithms with other algorithms for directed acyclic graphs and general directed graphs
%for both random and real graphs. 
%We then present a semi-streaming algorithm for building DFS tree of a random graph in Section~\ref{sec:semi-stream}.
In Section~\ref{sec:real_exp}, we evaluate the existing algorithms on real graphs and proposes simpler algorithms that 
perform very well on real graphs.
Finally, Section~\ref{sec:conclusion} presents some concluding remarks %about our work 
and the scope for future work.
%Finally, in Section~\ref{sec:conclusion} we present some concluding remarks about our work and scope for future work.
%Due to lack of space, some proofs have been omitted from the main text which can be found in Appendix~\ref{appn:proofs}.

%% file: prelim.tex
%\vspace{-1em}
\section{Preliminary}
\label{sec:prelim}
%\vspace{-.75em}
For all the experiments described in this paper, we add a pseudo root to the graph, i.e.,
a dummy vertex $s$ that is connected to all vertices in the graph $G$.
All the algorithms thus start with an empty graph augmented with the pseudo root $s$ and its edges,
and maintain a DFS tree rooted at $s$ after every edge insertion. 
%All the algorithms thus start with any arbitrary DFS tree $T$ rooted at 
%$s$ in this augmented graph and always maintain a DFS tree rooted at $s$ after every edge insertion.
It can be easily observed that each subtree rooted at any 
child of $s$ is a DFS tree of a connected component of the graph.
Given the graph $G$ under insertion of edges, %=(V,E)$ on $n=|V|$ vertices and $m=|E|$ edges,
the following notations will be used throughout the paper.
%%\vspace{-7pt}
%%\begin{itemize}
%$\bullet$ $T:$~ A DFS tree of $G$ at any time during the algorithm.
%%\item $s:$~ Root of the tree $T$.
%%\item $par(v):$ Parent of $v$ in $T$.
%
%$\bullet$ $path(x,y):$ Path from the vertex $x$ to the vertex $y$ in $T$.
%
%$\bullet$ $T(x):$ The subtree of $T$ rooted at a vertex $x$.
%
%$\bullet$ $LCA(u,v):$ The lowest common ancestor of $u$ and $v$ in $T$.\\
%%\item $LA(u,k):$ The ancestor of $u$ at level $k$ in tree $T$.
%%\end{itemize}
\begin{itemize}
\item $T:$~ A DFS tree of $G$ at any time during the algorithm.
%\item $s:$~ Root of the tree $T$.
%\item $par(v):$ Parent of $v$ in $T$.
\item $path(x,y):$ Path from the vertex $x$ to the vertex $y$ in $T$.
\item $T(x):$ The subtree of $T$ rooted at a vertex $x$.
\item $LCA(u,v):$ The lowest common ancestor of $u$ and $v$ in $T$.
%\item $LA(u,k):$ The ancestor of $u$ at level $k$ in tree $T$.
\end{itemize}
The two prominent models for studying random graphs are $G(n,m)$~\cite{Bollobas84} and $G(n,p)$~\cite{ErdosR59,ErdosR60}. 
A random graph $G(n,m)$ consists of the first $m$ edges of a uniformly random permutation of all possible 
edges in a graph with $n$ vertices. In a random graph $G(n,p)$, every edge is added to the graph with a probability of $p$
independent of other edges. 
Further, a graph property $\cal P$ is called a {\em monotone increasing} graph property if $G\in \cal P$ implies that
$G+e \in \cal P$, where $G+e$ represents the graph $G$ with an edge $e$ added to it.
We now state the following classical results for random graphs that shall be used in our analysis.
%\vspace{-.75em}
%NEEDS UPDATION
\begin{theorem}\cite{FriezeK15} %[Theorem 4.1 in \cite{FriezeK15}]
	Graph $G(n,m)$ with $m=\frac{n}{2}(\log n + c)$  is connected with probability at least $e^{-e^{-c}}$ 
	for any constant $c$.
\label{theorem:Frieze}
\end{theorem}
%\begin{theorem}
%	Graph $G(n,p)$ with $p=(\log n + c)/n$  is connected with probability at least $1-e^{-c}$ for any constant $c$.
%\label{theorem:Frieze}
%\end{theorem}
%\vspace{-1.25em}
\begin{theorem}\cite{FriezeK15}
	For a monotone increasing graph property $\cal P$ and $p=m/{n\choose 2}$, 
	$Pr[G(n,p)\in {\cal P}]$ is equivalent to $Pr[G(n,m)\in {\cal P}]$ up to small constant factors.
	\label{theorem:GnpGnm}
\end{theorem}

%
%\begin{theorem}[Erdos Renyl (1960)]
%	Given a random graph $G(n,m)$, the graph will become completely connected when the number of edges in the 
%	graph is $n\log n$ with probability $1-1/n^c$ for some constant $c>2$.
%	\label{thm:ERConn}
%\end{theorem}
%
%\begin{theorem}[Erdos Renyl (1960)]
%	Given a random graph $G(n,m)$, the graph enters a phase transition around $m=cn/2$ edges, where w.h.p.
%	\begin{enumerate}
%	\item For $c<1$,	 all the components of the graph have size $O(\log n)$.
%	\item For $c=1$, the largest component has the size $O(n^{2/3})$.
%	\item For $c>1$, there exist a single giant component of size $\Theta(n)$, with all the other components having size $O(\log n)$.	
%\end{enumerate}	 
%	\label{thm:ERGiant}
%\end{theorem}

%\vspace{-1.5em}
\subsection{Experimental Setting}
%\vspace{-.75em}
In our experimental study on random graphs, the performance of different algorithms 
is analyzed in terms of the number of edges processed, instead of the time taken. 
This is because the total time taken by an algorithm is dominated by the time taken to process the graph edges 
(see Appendix \ref{appn:perf_edge} for details). 
Furthermore, comparing the number of edges processed provides a deeper insight 
in the performance of the algorithm (see Section~\ref{sec:basic_exp}).
Also, it makes this study independent of the computing platform making it easier to reproduce and verify.
%Thus, these experiment independent of the computing platform making it easier to reproduce and verify.
For random graphs, each experiment is averaged over several test cases to get the expected behavior.
For the sake of completeness, the corresponding experiments are also replicated measuring the time taken by 
different algorithms in Appendix \ref{appn:time_plots}.
However, for real graphs the performance is evaluated by comparing the time taken and not the edges processed.
This is to ensure an exact evaluation of the relative performance of different algorithms.
%Each experiment on random graph is averaged over several test cases to get the expected behavior.

%\vspace{-1.25em}
\subsection{Datasets}
%\vspace{-.75em}
In our experiments we considered the following types of datasets.
%\itemsep-1em 
%\vskip{-1ex}
\begin{itemize}
%\noindent
%$\bullet$ 
\item
\textbf{Random Graphs: } The initial graph is the star graph, formed by adding an edge 
from the pseudo root $s$ to each vertex.
The update sequence is generated based on Erd\H{o}s R\'{e}nyi $G(n,m)$ model by choosing the first $m$ edges of a random 
permutation of all the edges in the graph. For the case of DAGs, the update sequence is generated based on extension
of $G(n,m)$ model for DAGs \cite{CordeiroMPTVW10}.
%\noindent
%$\bullet$ 
\item
\textbf{Real graphs: } 
We use a number of publically available datasets \cite{Konect13,snapDATA,bigDND14} derived from real world.
These include graphs related to Internet topology, collaboration networks, online communication, friendship networks and 
other interactions (for details refer to Section~\ref{apn:real_datasets}). 
%For more details about these datasets refer to Appendix~\ref{apn:real_datasets}. 
\end{itemize}

\input{existing_algos}

%% file: existing_algos.tex
\section{Existing algorithms}
\label{sec:exist_algos}
In this section we give a brief overview of the %two prominent 
results on maintaining incremental DFS. The key ideas used in these algorithms are crucial to 
understand their behavior on random graphs.

\subsubsection*{Static DFS algorithm (SDFS)}
The basic static algorithm for evaluating the DFS tree of a graph was given by Tarjan \cite{Tarjan72}.
In the incremental version of the same, SDFS essentially computes the whole DFS tree from scratch 
after every edge insertion.

\subsubsection*{Static DFS algorithm with interrupt (SDFS-Int)}
Static DFS tree was shown to have much better performance for a random graph by 
Kapidakis \cite{Kapidakis90}. Only difference from SDFS is that the algorithm terminates as 
soon as all the vertices of the graph are marked visited. Again, the algorithm recompute the DFS tree 
from scratch after every edge insertion though requiring only $O(n\log n)$ time for random graphs. 

%\subsection*{Incremental DFS algorithm for DAG (FDFS)}
\subsection*{Incremental DFS for DAG/directed graph (FDFS)}
%The first incremental algorithm for maintaining a DFS tree was given by Franciosa et al. \cite{FranciosaGN97}.
FDFS \cite{FranciosaGN97} maintains the post-order (or DFN) numbering of vertices in the DFS tree, which is used to 
rebuild the DFS tree efficiently.
%perform the rebuilding efficiently.
%The algorithm maintains the post-order (or DFN) numbering of the vertices in the DFS tree, which is used to perform the rebuilding efficiently.
On insertion of an edge $(x,y)$ in the graph, it first checks whether $(x,y)$  is an anti-cross edge by verifying 
if DFN$[x]<$DFN$[y]$. In case $(x,y)$ is not an anti-cross edge, it simply updates the graph and terminates.
Otherwise, it performs a partial DFS on the vertices reachable from $y$ in the subgraph induced by the vertices 
with DFN number between DFN$[x]$ and DFN$[y]$. In case of DAGs, this condition essentially represents a 
{\em candidate set} of vertices that lie %on $path(LCA(x,y),x)$ and 
in the subtrees hanging on the right of $path(LCA(x,y),x)$ or on the left of $path(LCA(x,y),y)$.
%subgraph induced by the vertices with DFN number between DFN$[x]$ and DFN$[y]$.
%This essentially represents a {\em candidate set} of vertices that lie on $path(LCA(x,y),x)$ and in the 
%subtrees hanging on the right of $path(LCA(x,y),x)$ and on the left of $path(LCA(x,y),y)$.
%This essentially represents a {\em candidate set} of vertices that lie on $path(x,LCA(x,y))$ and in the subtrees 
%hanging from the $path(x,y)$ internal to it.
%However, the vertices reachable from $y$ in this subgraph cannot be present on $path(LCA(x,y),x)$ and $path(LCA(x,y),y)$,
%as it will lead to a cycle on inclusion of $(x,y)$, which is not allowed in DAGs. 
FDFS thus removes these reachable vertices from the 
corresponding subtrees and computes their DFS tree rooted at $y$ to be hanged from the edge $(x,y)$. %adds them as a subtree from $y$.
%However, in case of DAGs, the vertices reachable in this subgraph from $y$ cannot be present on $path(x,y)$, 
%which will lead to a cycle on inclusion of $(x,y)$. Thus, FDFS removes the reachable vertices from the corresponding 
%trees and adds them as a subtree from $y$.
The DFN number of all the vertices in candidate set is then updated to perform the next insertion efficiently.
%Note that this updation of DFN numbers can be much efficiently using advanced data structures as Euler tour trees \cite{Tarjan97}.
%Thus, the time taken by the algorithm is proportional to the number of edges processed by it. 
The authors also noted that FDFS can be trivially extended to directed graphs.
Here, the {\em candidate} set includes the subtrees hanging on the right
of $path(LCA(x,y),x)$ until the entire subtree containing $y$ (say $T'$).
Note that for DAGs instead of entire $T'$, just the subtrees of $T'$ hanging on the left of 
$path(LCA(x,y),y)$ are considered. 
However, the authors did not give any bounds better than $O(m^2)$ for FDFS in directed graphs.

%\subsection*{Incremental DFS with amortized guarantee (ADFS1 and ADFS2)}
\subsection*{Incremental DFS for undirected graphs (ADFS)}
%The first incremental algorithm for maintaining a DFS tree for undirected graphs was presented by Baswana and Khan \cite{BaswanaK14}.
ADFS \cite{BaswanaK14} (refers to both ADFS1 and ADFS2) maintains a data structure that answers LCA and level ancestor
% (LA) 
 queries.
% The essential idea of ADFS (refers to both ADFS1 and ADFS2) revolves around two key principles 
% called as {\em minimal restructuring} and {\em monotonic fall}
% 
%In this paper they describe two algorithms, a basic algorithm (ADFS1) and an improved algorithm (ADFS2).
%The essential idea of ADFS (refers to both ADFS1 and ADFS2) revolves around two key principles 
%called as {\em minimal restructuring} and {\em monotonic fall}. 
%
On insertion of an edge $(x,y)$ in the graph, ADFS first verifies whether  $(x,y)$ is a cross edge
by computing $w= LCA(x,y)$ and ensuring that $w$ is not equal to either $x$ or $y$.
In case $(x,y)$ is a back edge, it simply updates the graph and terminates. Otherwise,
let $u$ and $v$ be the children of $w$ such that $x\in T(u)$ and $y\in T(v)$. Without loss of generality,
let $x$ be lower than $y$ in the $T$.  ADFS then rebuilds $T(v)$ hanging it from $(x,y)$ as follows.
It first reverses $path(y,v)$ which converts many back edges in $T(v)$ to cross edges.
It then collects these cross edges and iteratively inserts them back to the graph using the same procedure.
The only difference between ADFS1 and ADFS2 is the order in which these collected cross edges are processed.
ADFS1 processes these edges arbitrarily, whereas ADFS2 processes the cross edge with the highest endpoint first. 
For this purpose ADFS2 uses a non-trivial data structure. % , say $\cal D$.
We shall refer to this data structure as $\cal D$.
%Note that this algorithm requires a data structure to answer LCA queries, which can be performed 
%using advanced data structures as Euler tour trees \cite{Tarjan97} in $\tilde{O}(1)$ time. The remaining 
%time taken by the algorithm is proportional to the number of edges processed by it. 

\subsubsection*{Incremental DFS with worst case guarantee (WDFS)}
Despite several algorithms for maintaining DFS incrementally, 
the worst case time to update the DFS tree after an edge insertion was still $O(m)$. 
Baswana et al.~\cite{BaswanaCCK16} presented an incremental algorithm, 
giving a worst case guarantee of $O(n\log^3 n)$ on the update time. 
The algorithm builds a data structure using the current DFS tree, which
is used to efficiently rebuild the DFS tree after an edge update.
However, building this data structure requires $O(m)$ time and hence 
the same data structure is used to handly multiple updates ($\approx\tilde{O}(m/n)$). 
The data structure is then rebuilt over a period of updates using a technique called 
{\em overlapped periodic rebuilding}. 
Now, the edges processed for updating a DFS tree depends on the number of edges inserted since the 
data structure was last updated. Thus, whenever the data structure is updated, 
there is a sharp fall in the number of edges processed per update resulting in a saw like structure on 
the plot of number of edges processed per update.

%% file: basic_exp_random_graph.tex
\section{Experiments on Random Undirected graphs}
\label{sec:basic_exp}
%\vspace{-.75em}
We now compare the empirical performance of the existing algorithms for  incrementally 
maintaining a DFS tree of a random undirected graph. % described in Section \ref{sec:exist_algos}. 

%%%%%%%%%%%%%%%%%%%%%%%%%%%%%%%%%%%%%%%%%%%%
%\subsection{Fixed sparsity, Increasing n} %
%%%%%%%%%%%%%%%%%%%%%%%%%%%%%%%%%%%%%%%%%%%%

%\vspace{-.5em}
\newcommand{\figW}{0.325\linewidth}
\begin{figure*}[!ht]
\centering
\includegraphics[width=\figW]{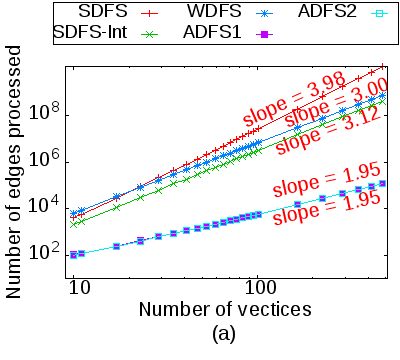}
\includegraphics[width=\figW]{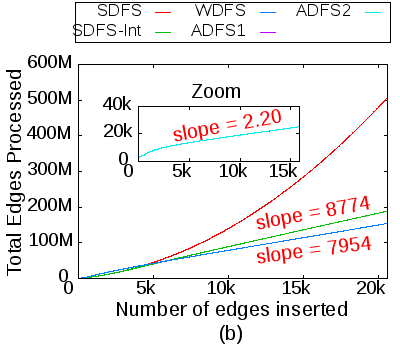}
\includegraphics[width=\figW]{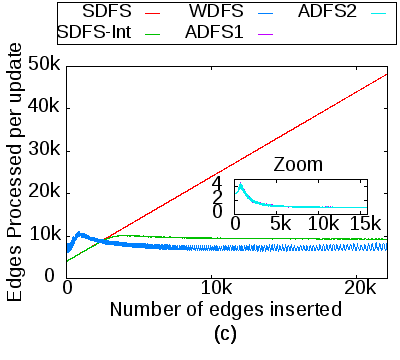}
\caption{
For various existing algorithms, the plot shows 
(a) Total number of edges processed (logarithmic scale) for insertion 
of $m={n \choose 2}$ edges for different values of $n$ ,
(b) Total number of edges processed for $n=1000$ and up to $n\sqrt{n}$ edge insertions, 
(c) Number of edges processed per update for $n=1000$ and up to $n\sqrt{n}$ edge insertions. 
See Figure \ref{fig:incDFSvarNMT} for corresponding time plot.
%(a) $m=2n$, (b) $m=n\log n$, (c) $m=n\sqrt{n}$, (d) $m={n \choose 2}$.
}
\label{fig:incDFSvarNM}
\end{figure*}

We first compare the total number of edges processed by the existing algorithms
for insertion of $m={n \choose 2}$ edges, as a function of number of vertices in
Figure \ref{fig:incDFSvarNM} (a).
Since the total number of edges is presented in logarithmic scale, the slope $x$ of a line 
depicts the growth of the total number of edges as $O(n^x)$.
%Figure \ref{fig:incDFSvarN} (a)  and Figure \ref{fig:incDFSvarN} (b) shows this comparison in normal 
%and logarithmic scale respectively.
%The slope $x$ of a line in Figure \ref{fig:incDFSvarNM} (b) depicts the variation of the corresponding curve in 
%Figure \ref{fig:incDFSvarN} (a) that grows as $O(n^x)$.
The performance of SDFS, SDFS-Int and WDFS resemble their asymptotic bounds described in Table \ref{tab:exist_algo}.
For small values of $n$, WDFS performs worse than SDFS and SDFS-Int because of large difference between the constant terms 
in their asymptotic bounds, which is evident from their y-intercepts. % in Figure \ref{fig:incDFSvarN} (b). 
%The larger constant term can be attributed to the simplicity of SDFS in comparison to WDFS. 
%This is also responsible for poorer performance of WDFS in comparison to SDFS-Int in spite of better asymptotic bounds (see Figure %\ref{fig:incDFSvarN} (b)).
However, the effect of constant term diminishes as the value of $n$ is increased. 
The most surprising aspect of this experiment is the exceptional performance of ADFS1 and ADFS2.
Both ADFS1 and ADFS2 perform extremely faster than the other algorithms.
%The most significant inference from this experiment is the performance of ADFS1 and ADFS2, which is far better than other algorithms. 
Furthermore, ADFS1 and ADFS2 perform equally well despite the  difference in their asymptotic complexity (see Table \ref{tab:exist_algo}).\\

{\centering
	\fbox{\parbox{\linewidth}{
			{\bf Inference $I_1$:}~ 
			ADFS1 and ADFS2 perform equally well and much faster than other algorithms. 
%			ADFS1 and ADFS2 perform equally well and much faster than other
%			existing algorithms on random graphs. 
		}}
}
%\vspace{-1em}
		%\vspace{-15pt}
%\end{center}

\noindent
\textbf{Remark:} 
Inference $I_1$ is surprising because the complexity of ADFS1 and ADFS2 has been shown \cite{BaswanaK14} to be 
$O(n^{3/2}\sqrt{m})$ and $O(n^2)$ respectively. 
Moreover, we present a sequence of $m$ edge insertions where ADFS1 takes 
$\Omega(n^{3/2}\sqrt{m})$ time in Appendix \ref{appn:wc_ADFS}, proving the tightness of its analysis. 
Further, ADFS2 takes slightly more time than ADFS1, in order to maintain the data structure
$\cal D$ (see Figure~\ref{fig:incDFSvarNMT}).\\

We now compare the total number of edges processed by the existing algorithms as a 
function of number of inserted edges in Figure \ref{fig:incDFSvarNM} (b).
The slopes of SDFS-Int, WDFS and ADFS represent the number of edges processed per edge insertion.
Here again, the performance of SDFS, SDFS-Int and WDFS resembles with their worst case values 
(see Table \ref{tab:exist_algo}). 
Similarly, both ADFS1 and ADFS2 perform equally well as noted in the previous experiment.
When the graph is sparse ($m<<n\log^3n$), WDFS performs worse than SDFS because of high cost of update 
per edge insertion (see Table~\ref{tab:exist_algo}). Further, as expected the plots of SDFS-Int and WDFS grow linearly in $m$.
This is because their update time per insertion is independent of $m$.
However, the plots of ADFS1 and ADFS2 are surprising once again, because they become almost linear as the graph becomes denser.
In fact, once the graph is no longer sparse, each of them processes  $\approx 2$ edges per edge insertion to maintain the DFS tree.
This improvement in the efficiency of ADFS1 and ADFS2 for increasing value of $m$ is 
counter-intuitive since more edges may be processed to rebuild the DFS tree as the graph becomes denser.\\

%We now compare the total number of edges processed by the existing algorithms as a function of number of inserted edges
%(refer to Figure \ref{fig:incDFSvarM} (a)).
%The slopes of SDFS-Int, WDFS and ADFS represent the number of edges processed per edge insertion.
%Here again, SDFS, SDFS-Int and WDFS give expected performance resembling the asymptotic bounds as described in 
%Figure \ref{fig:exist_algo}. 
%Similarly, both ADFS1 and ADFS2 perform equally as noted in the previous experiment.
%When the graph is sparse, WDFS performs worse than SDFS because of high cost of update 
%per edge insertion. Further, as expected the plots of SDFS-Int and WDFS are linear in $m$.
%This is because their update time per insertion is independent of $m$.
%However, the plots of ADFS1 and ADFS2 are surprising as they become almost linear as the graph becomes denser, 
%requiring to process  $\approx 2$ edges per edge insertion to maintain the DFS tree.
%This improvement in efficiency of update for increasing value of $m$ is counter-intuitive since more edges 
%may be processed to rebuild the DFS tree as the graph becomes denser.

%\vspace{.25em}
%\begin{center}
{\centering
	\fbox{\parbox{\linewidth}{
			{\bf Inference $I_2$:}~ 
			ADFS1/ADFS2 processes $\approx 2$ edges per insertion after the insertion of $O(n)$ edges.
%			ADFS1/ADFS2 processes $\approx 2$ edges per insertion after $O(n)$ edges are inserted in the graph.
		}}
}
%\vspace{-.75em}
		%\vspace{-15pt}
%	\end{center}

%Initially WDFS performs worse than SDFS which overcomes as the edges are increased.
%Since WDFS reqiures $\tilde{O}(n)$ time per update the edges processed increases almost linearly. 
%Also as SDFS requires $O(m)$ time per update the edges processed increases quadratically.
%ADFS1 and ADFS2 processes more number of edges when the graph is sparse, but becomes 
%surprisingly efficient as the graph becomes denser. 
%This anomalous behavior is further investigated in our next experiment.

%%%%%%%%%%%%%%%%%%%%%%%%%%%%%%%%%%%%%%%%%%%%%%%%%%%%%%%%%%%%%%%%%%%%%
% \subsection{Edges processed per insertion, Fixed n, Increasing m} %
%%%%%%%%%%%%%%%%%%%%%%%%%%%%%%%%%%%%%%%%%%%%%%%%%%%%%%%%%%%%%%%%%%%%%

%\renewcommand{\figW}{0.35\linewidth}
%
%\begin{figure}[h]
%\centering
%\includegraphics[width=\figW]{figures/IncDFSN_VarM_100_SP3_EPI.png}
%\includegraphics[width=\figW]{figures/IncDFS_VarM_1000_SP2_EPI.png}
%%\includegraphics[width=\figW]{figures/IncDFS_VarN_SP.png}
%\caption{Total number of edges processed per insertion by the existing algorithms starting from empty graphs 
%	on a random sequence of edge insertions for varying $m$, where $n=1000$.
%	(a) For 100 nodes for $m=n\sqrt{n}$ edge insertions, (b) For 1000 nodes for $m=n\log n$ edges insertions. 
%%(a) $m=2n$, (b) $m=n\log n$, (c) $m=n\sqrt{n}$, (d) $m={n \choose 2}$.
%}
%\label{fig:incDFSvarM_PI}
%\end{figure}

Finally, to investigate the exceptional behavior of ADFS1 and ADFS2, we compare the number of edges 
processed per edge insertion by the existing algorithms as a function of number of inserted edges 
in Figure \ref{fig:incDFSvarNM} (c).
Again, the expected behavior of SDFS, SDFS-Int and WDFS matches with their worst case bounds described in Table \ref{tab:exist_algo}. 
%Again, SDFS, SDFS-Int and WDFS presents their expected behavior matching the bounds 
%described in Table \ref{tab:exist_algo}. 
%SDFS-Int performs almost equal to SDFS for $m< o(n\log n)$ as expected following which it requires $O(n\log n)$ . 
The plot of WDFS shows the saw like structure owing to {\em overlapped periodic rebuilding} of the data structure 
used by the algorithm (see Section~\ref{sec:exist_algos}).
%\footnote{WDFS essentially uses a data structure which is updated over a period of several edge insertions. 
%The edges processed depends on the number of edges inserted since the data structure was last updated. 
%Thus, whenever the data structure is updated, there is a sharp fall in the number of edges processed per edge update
%resulting in a saw shaped plot.}
%This results in a saw like structure of the plot.}.
Finally, the most surprising result of the experiment is the plot of ADFS1 and ADFS2 shown in the zoomed component of the plot.
The number of edges processed per edge insertion sharply increases to roughly $5$ (for $n=1000$) when $m$ reaches $O(n)$ followed by a sudden fall to reach $1$ asymptotically. Note that the inserted edge is also counted among the processed edges, hence essentially the number of edges processed to update the DFS tree asymptotically reaches zero as the graph becomes dense.
%However, both ADFS1 and ADFS2 performs much better than expected, 
%explaining their performance as discussed above.
%The number of edges processed by ADFS1 and ADFS2 processes increases till $O(n)$ inserted edges after which 
%it decreases sharply to almost constant edges per insertion as the graph becomes denser. The maximum number of edges
%processed per insertion touches roughly $O(\log n)$.
This particular behavior is responsible for the exceptional performance of ADFS1 and ADFS2.\\

%\begin{center}
%\vspace{1.5em}
{\centering
	\fbox{\parbox{\linewidth}{
			{\bf Inference $I_3$:}~  
%			Edges processed per update by ADFS  asymptotically reaches zero in a dense graph.
%			No. of edges processed by ADFS1/ADFS2 per update reaches zero with as graph becomes dense.
			Number of edges processed by ADFS1/ADFS2 for updating the DFS tree \\asymptotically reaches zero as the graph becomes denser.
			}}
		%\vspace{-15pt}
}

%	\end{center}
%To understand $I_1$, $I_2$ and $I_3$, we investigate the structure of a DFS tree for random graphs as follows.
%% in the following section.

To understand the exceptional behavior of ADFS1 and ADFS2 for random graphs inferred in $I_1$, $I_2$ and $I_3$, 
we investigate the structure of a DFS tree for random graphs in the following section.

%% file: broomStick_exp.tex
\section{Structure of a DFS tree: The broomstick}
\label{sec:broomStick_exp}
%\vspace{-.5em}
We know that SDFS, SDFS-Int and WDFS invariably rebuild the entire DFS tree on insertion of every edge.
We thus state the first property of ADFS that differentiates it from other existing algorithms.\\

%We now study the structure of DFS tree for random graphs to uncover its properties
%responsible for the exceptional behavior of 

%%In this section we refer to both the algorithms ADFS1 and ADFS2 by ADFS. % (refers to both ADFS1 and ADFS2).
%%Consider the behavior of the existing algorithms on insertion of a {\em back edge}.
%%A key property responsible, is the behavior of the algorithm on insertion of a {\em back edge}.
%%All the existing algorithms except ADFS 
%The algorithms SDFS, SDFS-Int and WDFS invariably rebuild the entire DFS tree on insertion of every edge.
%%However, recall that ADFS implicitly avoids rebuilding the DFS tree if the inserted edge is a back edge
%%(see Section \ref{sec:exist_algos}) highlighting the following property of ADFS. 
%We now state the first property of ADFS that differentiates it from the other existing algorithms.\\
%% is a back edge
%%(see Section \ref{sec:exist_algos}) highlighting the following property of ADFS. 

%\begin{center}
%\vspace{.5em}
{\centering
	\fbox{\parbox{\linewidth}{
			{\bf Property $P_1$:}~ ADFS rebuilds the DFS tree only on insertion of a cross edge. 
		}}
}
%\textbf{\vspace{-1em}}

%\end{center}
\renewcommand{\figW}{0.35\linewidth}
\newcommand{\figWN}{0.365\linewidth}
\begin{figure*}[ht]
\centering
\includegraphics[width=\figW]{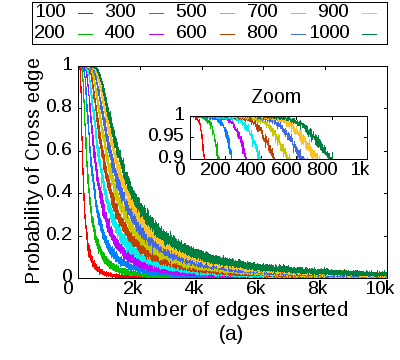}\qquad
\includegraphics[width=\figWN]{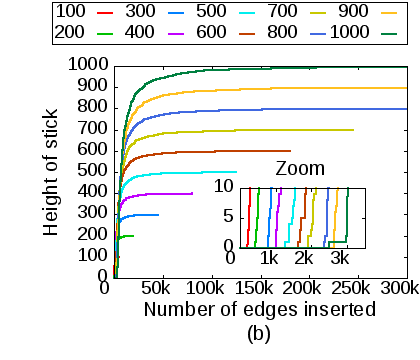}
\caption{
The variation of (a) $p_c:$ Probability of next inserted edge being a cross edge, and
(b) $l_s:$ Length of broomstick, with graph density. %for different values of $m$.
Different lines denote different number of vertices, starting at different points.% (see zoomed part). %values of $n$.
%Different number of vertices are shown by different lines, initiating at different points (see zoomed part).
%(a) Probability that next randomly inserted  edge is a cross edge ($p_c$),
%for different values of $m$.
%(b) Length of broomstick for different values of $m$.
%(a) Showing only for small values of n. Line is 3nlogn for some
%(b) Showing for even large values of n.
%(c) LogScale version.
%Different lines denote different number of vertices. %values of $n$.
%Zoomed portion shows the start of each line.
}
\label{fig:structDFS}
\end{figure*}
%%%%%%%%%%%%%%%%%%%%%%%%%%%%%%%%%%%%%%%%%%%%%
%% \subsection{Probability of Cross edge} %%%
%%%%%%%%%%%%%%%%%%%%%%%%%%%%%%%%%%%%%%%%%%%%%
Let $T$ be any DFS tree of the random graph $G(n,m)$.
Let $p_c$ denote the probability that the next randomly inserted edge is a cross edge in $T$. 
We first perform an experimental study to determine the behavior of $p_c$ as the number of edges in the graph increases.
Figure \ref{fig:structDFS} (a) shows this variation of $p_c$ for different values of $n$.
%Note that this variation of $p_c$ remains same when evaluated using different existing algorithms discussed in the previous section.
%
%Note that given $m$ and $n$, the value of $p_c$ is independent of the algorithm used to maintain the DFS tree.
%Hence, the variation of $p_c$ shown in Figure \ref{fig:structDFS} (a) is a property of DFS tree of a random graph 
%and not of any particular algorithm.
%We found that this probability is independent of the incremental algorithm used to maintain the DFS tree.
%Hence, it is a property of the DFS tree itself and not of the corresponding algorithm.
%Figure \ref{fig:structDFS} (a) shows the variation of probability that the newly inserted edge 
%is a cross edge for different values of $n$. 
The value $p_c$ starts decreasing sharply once the graph has $\Theta(n)$ edges.
Eventually, $p_c$ asymptotically approaches $0$ as the graph becomes denser.
Surely ADFS crucially exploits this behavior of $p_c$ in random graphs (using Property $P_1$).
%Clearly, this has a huge impact on the overall running time of ADFS. 
In order to understand the reason behind this behavior of $p_c$, 
we study the structure of a DFS tree of a random graph.
%Now, any newly inserted edge with atleast on endpoint on the stick necessarily becomes a back edges.
%Thus, the probability of insertion of a cross edges decreases sharply as the size of the stick increases.
%Fig \ref{fig:structDFS} (b) describes the variation of probability that the newly inserted edges is a cross edges 
%for different values of $n$. Since SDFS and WDFS rebuild the DFS tree even on insertion of back edges, 
%they do not exploit this property. Whereas, it is inherently exploited by ADFS1 and ADFS2.
%
%
%
%To understand the exceptional behavior of ADFS (refers to both ADFS1 and ADFS2) on random graph,
%we need to carefully examine the structure of DFS tree for random graph.
%ADFS implicitly exploits 
%\vspace{-1.25em}
\subsection*{Broomstick Structure}
%\vspace{-.75em}
%We now define a few terms that shall be used to describe the structure of a DFS tree henceforth.
The structure of a DFS tree can be described as that of a broomstick as follows. From the root of 
the DFS tree there exists a downward path on which there is no branching, i.e., every vertex has exactly one child. 
We refer to this path as the {\em stick} of the broomstick structure. The remaining part of the DFS tree (except the {\em stick})
is called the {\em bristles} of the broomstick. 
%Similar structure was studied by Sibeyn \cite{Sibeyn01}, but 
%without the explicit restriction of the absence of branching on the {\em stick}. We shall shortly see that this restriction 
%turns out to be very crucial for our results.

%DEFINE BROOM STICK AND ITS CORRRESPONDING TERMS.

%\subsection{Experiments}

%%%%%%%%%%%%%%%%%%%%%%%%%%%%%%%%%%%%%%%%%%%%%
%% \subsection{Length of broom} %%%
%%%%%%%%%%%%%%%%%%%%%%%%%%%%%%%%%%%%%%%%%%%%%

Let $l_s$ denote the length of the {\em stick} in the broomstick structure of the DFS tree.
We now study the variation of $l_s$ as the edges are inserted in the graph.
Figure \ref{fig:structDFS} (b) shows this variation of $l_s$ for different values of $n$.
%Again, note that this variation of $l_s$ remains same when evaluated using different existing algorithms discussed in the previous section.
%Again, we find that given $n$ and $m$, the value of $l_s$ is independent of the algorithm used to maintain the DFS tree.
%Hence, the variation of $l_s$ shown in Figure \ref{fig:structDFS} (b) is also a property of the DFS tree and not the algorithm used to build it.
%For different algorithms this length came out to be the same, indicating that length of the broomstick 
%is a property of DFS tree and not of the corresponding algorithm to build the DFS tree.
Notice that the {\em stick} appears after the insertion of roughly $n\log n$ edges (see the zoomed part of Figure \ref{fig:structDFS} (b)). 
%This can be explained by the connectivity threshold for random graphs 
%(refer to Theorem \ref{thm:ERConn}). Until the graph becomes connected (till $\Theta(n\log n)$ edges), 
%each component hangs as a separate subtree from the pseudo root $s$, limiting the value of $l_s$ to $0$. 
After that $l_s$ increases rapidly to reach almost $90\%$ of its height within just 
$\approx 3n\log n$ edges, followed by a slow growth asymptotically approaching its maximum height only near 
$O(n^2)$ edges.
%After that $l_s$ increases rapidly to reach almost $90\%$ of its height within just 
%$\approx 3n\log n$ edges. The growth of $l_s$ then declines asymptotically approaching its 
%maximum height only near $O(n^2)$ edge insertions.
Since any newly inserted edge with at least one endpoint on the stick necessarily becomes a back edge,
%the sharp increase in $l_s$ is responsible for the sharp decrease in $p_c$.
the sharp decrease in $p_c$ can be attributed to the sharp increase in $l_s$.
%Now, any newly inserted edge with at least one endpoint on the stick necessarily becomes a back edge.
%Thus, the reason behind the sharp decrease in $p_c$ can be attributed to the sharp increase in $l_s$.
We now theoretically study the reason behind the behavior of $l_s$ using properties of random graphs,
proving explicit bounds for $l_s$ described in Theorem \ref{thm:bristle-sizeM}. 
%prove explicit bounds for determining the value of $l_s$ described in Theorem \ref{theorem:broomstick-result} 
%using properties of random graphs.

%We now prove explicit bounds for determining the value of $l_s$ described in Theorem \ref{theorem:broomstick-result} 
%using properties of random graphs.

%% file: probability_analysis.tex
%\vspace{-1.25em}
\subsection{Length of the stick}
\label{sec:broomStick_th}
%\vspace{-.5em}
%We shall now prove Theorem \ref{theorem:broomstick-result}. 
The appearance of broomstick after insertion of $n\log n$ edges as shown in Figure \ref{fig:structDFS} (b) 
can be explained by the connectivity threshold for random graphs (refer to Theorem \ref{theorem:Frieze}). 
Until the graph becomes connected (till $\Theta(n\log n)$ edges), 
each component hangs as a separate subtree from the pseudo root $s$, limiting the value of $l_s$ to $0$. 
%This can be explained by the connectivity threshold for random graphs 
%(refer to Theorem \ref{thm:ERConn}). Until the graph becomes connected (till $\Theta(n\log n)$ edges), 
%each component hangs as a separate subtree from the pseudo root $s$, limiting the value of $l_s$ to $0$. 
%In this analysis we have not attempted to optimize the constants in our expressions in order to 
%provide a simple and clean exposition of the proof.
%To provide a simple and clean exposition of the proof, we have not attempted to optimize the constants in our expressions. 
To analyze the length of $l_s$ for $m=\Omega(n\log n)$ edges, we first prove a succinct bound on the probability of existence of a long path without branching during a DFS traversal in 
$G(n,p)$ in the following lemma.

%First we state the following lemma that present a succinct bound on the probability for the existence of a long path without branching during a DFS traverrsal in $G(n,p)$.

%\begin{lemma}
%Let $n_0\le n$ be any integer. Consider random graph $G(n,p)$ with $p=(\log n_0 + c)/n_0$. There exists a
%DFS path without branching of length $\ge n- n_0$ with probability at least $1-2/e^c$.
%\label{lemma:dfs-path}
%\end{lemma}
%\vspace{-.5em}
\begin{lemma}
Given a random graph $G(n,p)$ with $p=(\log n_0 + c)/n_0$, for any integer $n_0\le n$ and $c\geq 1$, 
there exists a path without branching of length at least $n- n_0$ in the DFS tree of $G$ with probability at least $1-2/e^c$.
\label{lemma:dfs-path}
\end{lemma}

\begin{proof}
Consider any arbitrary vertex $u=x_1$, the DFS traversal starting from $x_1$ continues along a path
without branching so long as the currently visited vertex has at least one unvisited neighbor. 
Let $x_j$ denotes the $j^{th}$ vertex visited during the DFS on $G(n,p)$ starting from $x_1$. 
The probability that $x_j$ has at least one neighbor in the unvisited graph is $1- (1-p)^{n-j}$. 
We shall now calculate the probability that $\langle x_1,\ldots, x_{n-n_0} \rangle$ is indeed a path. 
Let $v=x_{n-n_0}$. We partition this sequence from $v$ towards $u$ into contiguous subsequences such that 
the first subsequence has length $n_0$ and $(i+1)^{th}$ subsequence has length $2^i n_0$ (see Figure \ref{fig:broomE} (a)). 
%The probability of occurrence of a path corresponding to the first subsequence is
%\[
%\left( 1- \left(1-\frac{\log n_0 + c}{n_0}  \right)^{n_0} \right)^{n_0} \ge \left( 1 - \frac{1}{n_0 e^c} \right)^{n_0} \ge 1 - \frac{1}{e^c}
%\]
The probability of occurrence of a path corresponding to the $i^{th}$ subsequence is at least \\
%\begin{align*}
%\vspace{-.5em}
\centerline{$
\left( 1- \left(1-\frac{\log n_0 + c}{n_0}  \right)^{2^in_0} \right)^{2^i n_0} \ge 
\left( 1 - \left(\frac{1}{n_0e^c} \right)^{2^i}  \right)^{2^i n_0} \ge 1 - \left( \frac{1}{e^c} \right)^{2^i}
$}\\
%\vspace{-.25em}
%\end{align*}
%\[
%\left( 1- \left(1-\frac{\log n_0 + c}{n_0}  \right)^{2^in_0} \right)^{2^i n_0} \ge 
%\left( 1 - \left(\frac{1}{n_0e^c} \right)^{2^i}  \right)^{2^i n_0} \ge 1 - \left( \frac{1}{e^c} \right)^{2^i}
%\]
Hence, the probability that DFS from $u$ traverses a path of length $n-n_0$ is at least
%\begin{eqnarray*}
%\[ 
$\Pi_{i=0}^{\log_2 n} \left(1 - \frac{1}{t^{2^i}} \right)$ % \]
%\end{eqnarray*}
for $t=e^c$. The value of this expression is lower bounded by  $1-\frac{2}{e^c}$
using the inequality \newline 
$\Pi_{i=0}^{i=\log_2 t} \left(  1-\frac{1}{t^{2^i}}\right) > 1-\frac{2}{t}$, 
that holds for every $c\geq 1$ since it implies $t>2$.
%that holds for every $t>2$.
\end{proof}
%\vspace{-.25em}
In order to establish a tight bound on the length of {\em stick}, we need to choose the smallest value of $n_0$ 
such that once we have a DFS path of length $n-n_0$ without branching, the subgraph induced by the remaining $n_0$ vertices 
is still connected. If the number of edges in this subgraph is $m_0$, 
according to Theorem \ref{theorem:Frieze}, $n_0$ has to satisfy the following inequality.\\ 
\qquad \centerline{$m_0 = \frac{n_0^2}{n^2}m \geq \frac{n_0}{2} (\log n_0 +c)$}\\
%the following inequality.
%In order to establish the existence of the broomstick structure, we need to pick the smallest $n_0$ such that once we have a DFS path of length $n-n_0$ without branching, the subgraph induced by the remaining $n_0$ vertices is still a connected subgraph. Using Theorem \ref{theorem:Frieze}, $n_0$ has to satisfy the following equality.
%\vspace{-.25em}
%\[m_0 = \frac{n_0^2}{n^2}m \geq n_0 (\log n_0 +c)\]
%\vspace{-.25em}
It follows from Theorem \ref{theorem:Frieze} that the subgraph on $n_0$ vertices will be connected with probability at least $e^{-e^{-c}}\geq 1-e^{-c}$.
% In order to establish our result on the broomstick structure, we need to show that a DFS tree of $G(n,m)$ will look like a path of length $n-%n_0$ with a connected subgraph on $n_0$ vertices hanging from the other end. See Figure for its illustration. 
Combining this observation with Lemma \ref{lemma:dfs-path} and Theorem \ref{theorem:GnpGnm}, 
it follows that the probability that DFS tree of $G(n,m)$ 
is a broomstick with stick length $\ge n-n_0$ is at least
%\[\left(1-\frac{1}{e^c}\right) \left( 1-\frac{2}{e^c}\right) \ge 
$1-\frac{c_1}{e^c}$ (for a small constant $c_1$). This probability tends to 1 for any increasing function $c(n)$, 
where $c(n)\geq 1$ for all $n$. 
% to compare the predicted length of broom stick and probability of cross edge
%with its experimentally evaluated value.
%
%
We would like to state the following corollary that will be used later on.
\begin{corollary}
%Let $m=2^i \log n$. A random graph $G(n,m)$ will have bristles of size at most $n/2^i$ with probability $O(1/n)$.
For any random graph $G(n,m)$ with $m=2^i n \log n$, its DFS tree will have bristles of size at most $n/2^i$ with probability $1-O(1/n)$.
\label{cor:bristle-size}
\end{corollary}

%In order t
To demonstrate the tightness of our analysis we compare the length of the stick as predicted theoretically (for $c=1$) 
with the length  determined experimentally in Figure \ref{fig:broomE} (b), which is shown to match exactly. 
This phenomenon emphasizes the accuracy and tightness of our analysis.
%Notice, that the experimental and theoretical values of $l_s$ match exactly 
%emphasizing the accuracy and tightness of our analysis.
%Notice, that this length matches exactly with the length 
%predicted emphasizing the tightness of our analysis.

\renewcommand{\figW}{0.4\linewidth}
\renewcommand{\figWN}{0.4\linewidth}
\begin{figure}
\centering
\includegraphics[width=\figW]{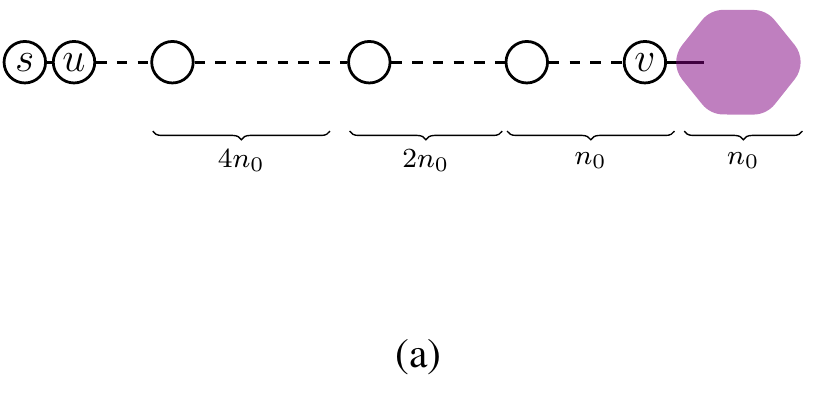}\qquad
\includegraphics[width=\figWN]{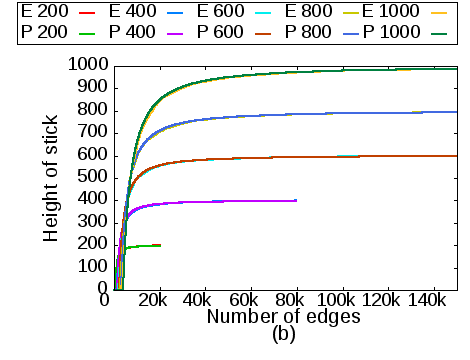}
\caption{ Estimating the length of 'stick' in the DFS tree.
(a) Partitioning the stick. 
(b) Comparison of experimentally evaluated (E) and theoretically predicted (P) value of  
length of the stick in the broomstick structure for different number of vertices.
}
\label{fig:broomE}
\end{figure}

%% file: broomStick_app.tex
\input{broomStickAppFull}

\section{New algorithms for Random Graphs}
\label{sec:new_algos}
%\vspace{-.75em}
Inspired by Lemma \ref{lemma:bristle-edges} and Lemma \ref{lem:bristle-algo} we propose the following 
new algorithms.
%\vspace{-1.25em}
\subsubsection*{Simple variant of SDFS (SDFS2) for random undirected graphs} 
%\vspace{-.25em}
We propose a {\em bristle oriented} variant of SDFS which satisfies the properties $P_1$ and $P_2$ of ADFS, i.e.,
it rebuilds only the bristles of the DFS tree on insertion of only cross edges.
This can be done by marking the vertices in the bristles as unvisited and performing the DFS traversal 
from the root of the bristles. Moreover, we also remove the non-tree edges incident on the {\em stick} of the 
DFS tree.
As a result, SDFS2 would process only the edges in the bristles, making it {\em bristle oriented}.
%Otherwise, SDFS2 would process all the edges of a vertex in the bristles (even those terminating on the stick),
%ceasing to be {\em bristle oriented}.
Now, according to Lemma \ref{lem:bristle-algo} the time taken by SDFS2 for insertion of $m= 2n\log n$ edges 
(and hence any $m\leq {n \choose 2}$) is $O(m^2)=O(n^2\log^2 n)$. 
%
%Moreover, this algorithm also works for directed graphs with similar bounds because 
%it is easy to show that Lemma \ref{lem:bristle-algo} holds for directed graphs as well.
%The significance of this algorithm is highlighted by the fact that there does not exists any non-trivial
%algorithm for maintenance of Incremental DFS for general directed graphs requiring $o(m^2)$ time. 
Thus, we have the following theorem.
%t\vspace{-.25em}
\begin{theorem}
Given a random graph $G(n,m)$, the expected time taken by SDFS2 for maintaining a DFS tree of $G$ 
incrementally is $O(n^2\log^2 n)$.
\label{theorem:theoretical-sdfs2}
\end{theorem}
%\vspace{-.25em}

%\renewcommand{\figW}{0.39\linewidth}
%\begin{figure*}[h]
%	\centering
%	\includegraphics[width=\figW]{figures/IncDFSN_VarN_SP4L_SDFS1N.png}
%	%\includegraphics[width=\figW]{figures/IncDFS_VarN_SP1L.png}
%	%\includegraphics[width=\figW]{figures/IncDFS_VarN_SP2L.png}
%	%\includegraphics[width=\figW]{figures/IncDFS_VarN_SP3L.png}
%%	\includegraphics[width=\figW]{figures/IncDFSN_VarN_SP4L.png}
%	\includegraphics[width=\figW]{figures/IncDFSN_VarM_1000_SP3_SDFS1_EPIN.png}
%	%\includegraphics[height=33cm]{figures/plot_vary_static_logFit.png}
%	\caption{Comparison of existing and proposed algorithms on undirected graphs:
%		(a) Total number of edges processed for insertion 
%		of $m={n \choose 2}$ edges for different values of $n$ in logarithmic scale 
%		%(fitting curve to $m+n^{\alpha}$).
%		(b)	Number of edges processed per edge insertion 
%		for $n=1000$ and up to $n\sqrt{n}$ edge insertions. 	
%		%(a) $m=2n$, (b) $m=n\log n$, (c) $m=n\sqrt{n}$, (d) $m={n \choose 2}$.
%		See Figure \ref{fig:incDFS-propT} for corresponding time plot.
%	}
%	\label{fig:incDFS-prop}
%\end{figure*}

We now compare the performance of the proposed algorithm SDFS2 with the existing algorithms. 
Figure \ref{fig:incDFS-prop} (a) compares the total number of edges processed for insertion of $m={n \choose 2}$ edges, 
as a function of number of vertices in the logarithmic scale. 
%Since the number of edges $m$ overshadows the true bounds of SDFS4 and ADFS, the slope of the line is evaluated
%to after removing $m$ edges representing the function $m+O(n^\alpha)$.
%This presents the performance of ADFS and SDFS4 in the right perspective. 
%As expected SDFS2 processes $m+\Omega(n^2)$ edges whereas both SDFS4 and ADFS processes $m+o(n^2)$ edges.
As expected SDFS2 processes $\tilde{O}(n^2)$ edges similar to ADFS. %whereas both SDFS4 and ADFS processes $m+o(n^2)$ edges.
%The predicted complexity of SDFS2, SDFS4 and ADFS matches the performance measured experimentally.
%Moreover, SDFS4 performs even better than ADFS for very high values of $n$. 
Figure \ref{fig:incDFS-prop} (b) compares the number of edges processed per edge insertion 
as a function of number of inserted edges.
%Notice the huge difference between SDFS2 and ADFS and equally huge difference between SDFS2 and ADFS around the insertion of $n/2$ edges.
Again, as expected SDFS2 performs much better than WDFS and SDFS-Int, 
performing asymptotically equal to ADFS as the performance differs only when the graph is very sparse ($\approx n\log n$).
Interestingly, despite the huge difference in number of edges processed by SDFS2 with respect to ADFS 
(see Figure \ref{fig:incDFS-prop} (a)), it is faster than ADFS2 and equivalent to ADFS1 in practice 
(see Figure \ref{fig:incDFS-propT} (a)).
%owing to its sheer simplicity.
%. This can be attributed to the sheer simplicity of SDFS2.
%but still lacks behind SDFS4 by a huge margin.
%The same is true for SDFS4, which performs much better than SDFS2 but still lacks behind ADFS by a huge margin.
%between SDFS2 and SDFS4 and equally huge difference between SDFS4 and ADFS around the insertion of $n/2$ edges.

\renewcommand{\figW}{0.34\linewidth}
\begin{figure}[h]
	\centering
	\includegraphics[width=\figW]{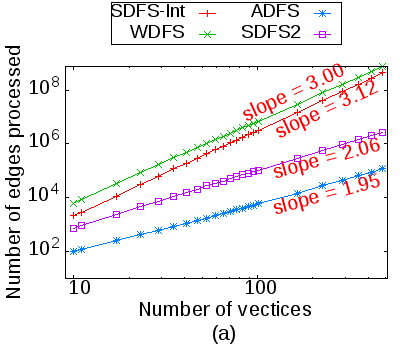}\qquad
	\includegraphics[width=\figW]{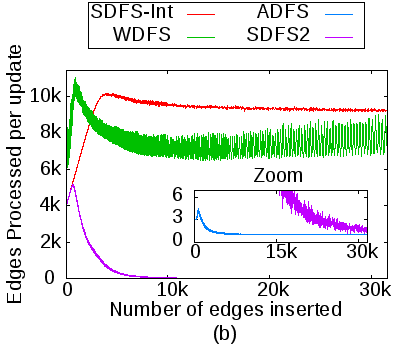}
	\caption{Comparison of existing and proposed algorithms on undirected graphs:
		(a) Total number of edges processed for insertion 
		of $m={n \choose 2}$ edges for different values of $n$ in logarithmic scale 
		%(fitting curve to $m+n^{\alpha}$).
		(b)	Number of edges processed per edge insertion 
		for $n=1000$ and up to $n\sqrt{n}$ edge insertions. 	
		%(a) $m=2n$, (b) $m=n\log n$, (c) $m=n\sqrt{n}$, (d) $m={n \choose 2}$.
		See Figure \ref{fig:incDFS-propT} for corresponding time plot.
	}
	\label{fig:incDFS-prop}
\end{figure}

\subsection*{Experiments on directed graphs and directed acyclic graphs} 
\label{appn:ExtDDAG}
\input{exp_directed_dag}

Finally, Lemma \ref{lemma:bristle-edges} also inspires the following interesting applications of 
SDFS2 in the semi-streaming environment 
as follows.
%(refer to Appendix \ref{appn:semi-stream} for details).
%%\vspace{-.25em}
%\begin{theorem}
%Given a random graph $G(n,m)$, there exists a single pass semi-streaming algorithm for maintaining
%DFS
%%, topological ordering 
%and strong connectivity 
%	%, strong bi-connectivity and strong 2-edge connectivity
%of $G$ incrementally using $O(n\log n)$ space.
%	%
%	%There exists a single pass semi-streaming algorithm for  incrementally maintaining a 
%	%data structure that answers queries on topological ordering, strong connectivity, strong articulation points and 
%	%strong bridges, requiring $O(n\log n)$ space for random graphs on $n$ vertices.
%	% and total $O(m + n^2\log^2 n)$ time for a stream of edge insertions.
%	\label{theorem:semi-streaming-apps}
%\end{theorem}

\input{semi_stream}

%% file: broomStickAppFull.tex
\section{Implications of broomstick property}
\label{sec:BroomStickImp}
%
%\subsection{Predicting Length of Broom}
%
%\subsection{Size of Bristles}
%\begin{figure}[ht!]
%\centering
%\includegraphics[height=4.5cm]{figures/plot_vary_bristles.png}
%\includegraphics[height=4.5cm]{figures/plot_vary_bristles_new.png}
%\caption{Figures Demonstrating the number of bristles versus number of edges in the graph. 
%Different lines for different number of vertices.
%Zoomed portion shows the start of each line.
%(a) Showing in log scale the number of bristles left.
%(b) Added line fit.
%}
%\label{fig:broom2}
%\end{figure}
%
%\textbf{Inferences from Figure \ref{fig:broom2} }
%\begin{enumerate}
%\item Prediction based on probability analysis predicts close to linear
%\end{enumerate}
%
%
%\subsection{Average time complexity of ADFS1}
Though the broomstick structure of DFS tree was earlier studied by Sibeyn \cite{Sibeyn01}, 
the crucial difference in defining the `stick' to be without branches proved to be extremely significant.
To emphasize its significance we now present a few applications of the broomstick structure of DFS tree, 
in particular Corollary \ref{cor:bristle-size} to state some interesting results. 
%(see Appendix \ref{appn:proofs} for proofs). 
Note that the absence of branches on the stick is extremely crucial for all of the following applications.
%(see Appendix~\ref{appn:proofs} for proofs).
%As an application of the broomstick structure of DFS tree, in particular Corollary \ref{cor:bristle-size} we shall  
%show the following interesting result. 

%\newtheorem*{LemBristle}{\textbf{Lemma~\ref{lemma:bristle-edges}}}
%\begin{LemBristle}
\begin{lemma}
For a uniformly random sequence of edge insertions, the number of edge insertions with both endpoints in bristles of
the DFS tree will be $O(n\log n)$ 
\label{lemma:bristle-edges}
\end{lemma}
%\end{LemBristle}
\begin{proof}
	We split the sequence of edge insertions into phases and analyze the expected number of edges inserted 
	in bristles in each phase. In the beginning of first phase there are $n\log n$ edges. In the $i^{th}$ phase, 
	the number of edges	in the graph grow from $2^{i-1}n\log n$ to $2^in\log n$. 
	%For notational convenience, in stage $0$ the edges grows from $0$ to $n\log n$.
	It follows from Corollary \ref{cor:bristle-size} that  $n_i$, the size of bristles in the $i^{th}$ phase will be
	at most $n/2^{i-1}$ with probability $1-O(1/n)$. Notice that each edge inserted during $i^{th}$ phase will 
	choose its endpoints randomly uniformly. Therefore, the expected number of edges with both endpoints in bristles 
	during the $i$th phase would be 
	\begin{equation*}	
	 m_i = \frac{n_i^2}{n^2}m \leq 2^in\log n/2^{2(i-1)} = n\log n/2^{i-2}	
	\end{equation*}		
%	\[ m_i = \frac{n_i^2}{n^2}m \leq 2^in\log n/2^{2(i-1)} = n\log n/2^{i-2}\]
	Hence, the expected number of edges inserted with both endpoints in bristles is $\sum_{i=1}^{\log n} m_i =O(n\log n)$.
\end{proof}

%%% COMMENTED ON JULY 11
%\begin{proof}
%We split the sequence of edge insertions into phases and analyze the expected number of edges inserted in bristles 
%in each phase. In the beginning of first phase there are $n\log n$ edges. In the $i^{th}$ phase, the number of edges 
%in the graph grow from $2^{i-1}n\log n$ to $2^in\log n$. 
%%For notational convenience, in stage $0$ the edges grows from $0$ to $n\log n$.
%It follows from Corollary \ref{cor:bristle-size} that  $n_i$, the size of bristles in the $i^{th}$ phase will be
%at most $n/2^i$ with probability $1-O(1/n)$. Notice that each edge inserted during $i^{th}$ phase will choose its endpoints
%randomly uniformly. Therefore, the expected number of edges with both endpoints in bristles during the $i$th phase would be \[ m_i = \frac{n_i^2}{n^2}m \leq 2^in\log n/2^{2(i-1)} = n\log n/2^{i-2}\]
%Hence the expected number of edges inserted with both endpoints in bristles is $O(n\log n)$.
%\end{proof}

In order to rebuild the DFS tree after insertion of a cross edge, it is sufficient to rebuild only the bristles of the 
broomstick, leaving the {\em stick} intact (as cross edges cannot be incident on it).
Corollary \ref{cor:bristle-size} describes that the size of bristles decreases rapidly as the graph becomes denser
making it easier to update the DFS tree.
This crucial insight is not exploited by the algorithm SDFS, SDFS-Int or WDFS.
We now state the property of ADFS that exploits this insight implicitly.\\
%However, this feature is implicitly exploited by ADFS  avoids rebuilding the {\em stick} giving an important property as follows:
%Once the broomstick structure is formed, the newly insno the {\em stick} can never have a cross edge on it.
%on insertion of any edge 
%which updates only the bristles of the tree giving the following property.
%Now, recall that on insertion of an edge $(x,y)$, ADFS rebuilds a child subtree of the $LCA(x,y)$ (refer to Section \ref{sec:exist_algos}).
%Thus, once the broomstick structure is formed, the {\em stick} will never be modified by ADFS on insertion of any edge 
%which updates only the bristles of the tree giving the following property.

%\begin{center}
{\centering
	\fbox{\parbox{\linewidth}{
			{\bf Property $P_2$:}~ ADFS modifies only the bristles of the DFS tree keeping the stick intact. 
		}}
}
%\end{center}

We define an incremental algorithm for maintaining a DFS for random graph to be {\em bristle oriented}
if executing the algorithm $\cal A$ on $G$ is equivalent to executing the algorithm on the subgraph induced by the bristles.
Clearly, ADFS is bristle oriented owing to property $P_2$ and the fact that it processes only the edges with both endpoints in rerooted subtree
(refer to Section \ref{sec:exist_algos}). 
We now state an important result for any  bristle oriented algorithm (and hence ADFS) as follows. 
% (see Appendix \ref{appn:proofs} for proof). 
%to process $m$ edge insertions is asymptotically equal to the time required to process the
% first $2n\log n$ edge insertions.
%
%for random graphs the total time required by any  bristle oriented algorithm (and hence ADFS) 
%to process $m$ edge insertions is asymptotically equal to the time required to process the
% first $2n\log n$ edge insertions.

%\newtheorem*{LemBristleAlgo}{\textbf{Lemma~\ref{lem:bristle-algo}}}
%\begin{LemBristleAlgo}
\begin{lemma}
	For any bristle oriented algorithm $\cal A$ if the expected total time taken to insert the first $2n\log n$ edges 
	of a random graph is $O(n^\alpha \log^\beta n)$ (where $\alpha>0$ and $\beta\geq 0$), 
	the expected total time taken to process any sequence of $m$ edge insertions is $O(m+n^\alpha \log^\beta n)$.
	\label{lem:bristle-algo}
\end{lemma}
%\end{LemBristleAlgo}

%\begin{lemma}
%	For any bristle oriented algorithm $\cal A$ if the expected total time taken to insert the first $2n\log n$ edges 
%	of a random graph is $O(n^\alpha \log^\beta n)$ (where $\alpha>0$ and $\beta\geq 0$), 
%	the expected total time taken to process any sequence of $m$ edge insertions is $O(m+n^\alpha \log^\beta n)$.
%	\label{lem:bristle-algo}
%\end{lemma}
\begin{proof}
	Recall the phases of edge insertions described in the proof of Lemma \ref{lemma:bristle-edges}, 
	where in the $i^{th}$ phase the number of edges in the graph grow from $2^{i-1} n\log n$ to $2^i n \log n$.
	The size of bristles at the beginning of $i^{th}$ phase is $n_i=n/2^{i-1}$  w.h.p..
	Further, note that the size of bristles is reduced to half during the first phase, 
	and the same happens in each subsequent phase  w.h.p. (see Corollary \ref{cor:bristle-size}). 
	%which happens similarly in all the subsequent phases. 
	Also, the expected number of edges added to subgraph represented by the bristles in $i^{th}$ phase is $O(n_i\log n_i)$
	(recall the proof of Lemma \ref{lemma:bristle-edges}). 
	Since $\cal A$ is bristle oriented, it will process only the subgraph induced by the bristles of size $n_i$
	in the $i^{th}$ phase.
	Thus, if $\cal A$ takes $O(n^\alpha \log^\beta n)$ time in first phase, the time taken by $\cal A$ in the 
	$i^{th}$ phase is $O(n_i^\alpha \log^\beta n_i)$. %(NEED MORE ARGUMENTS??).
	The second term $O(m)$ comes from the fact that we would need to process each edge to check whether it
	lies on the stick. This can be easily done in $O(1)$ time by marking the vertices on the stick. 
	The total time taken by $\cal A$ is $O(n^\alpha \log^\beta n)$ till the end of the first phase and 
	in all subsequent phases is given by the following  
	%\begin{align*}
	%	&= m+ \sum_{i=2}^{\log n} c n_i^\alpha \log^\beta n_i \qquad \leq \sum_{i=2}^{\log n} c (n/2^i)^\alpha \log^\beta (n/2^i) \\
	%	&\leq m+c n^\alpha \log^\beta n \sum_{i=2}^{\log n} 1/2^{i\alpha} \qquad (\log^\beta (n/2^i) \leq \log n, \text{for $\beta\geq 0$}) \\
	%	&\leq m+ c*c' n^\alpha \log^\beta n  \qquad (\sum_{i=2}^{\log n} 1/2^{i\alpha} = c', \text{for $\alpha>0$}) 
	%\end{align*}
	\begin{align*}
	&= m+ \sum_{i=2}^{\log n} c n_i^\alpha \log^\beta n_i 
		&\leq \sum_{i=2}^{\log n} c \Big(\frac{n}{2^{i-1}}\Big)^\alpha \log^\beta \Big(\frac{n}{2^{i-1}}\Big) \\
	&\leq m+c n^\alpha \log^\beta n \sum_{i=2}^{\log n} \frac{1}{2^{(i-1)\alpha}} &(\text{for $\beta\geq 0$}) \\
	&\leq m+ c*c' n^\alpha \log^\beta n   &(\sum_{i=2}^{\log n} \frac{1}{2^{(i-1)\alpha}} = c', \text{for $\alpha>0$}) 
	\end{align*}
	Thus, the total time taken by $\cal A$ is $O(m+ n^\alpha \log^\beta n)$.
\end{proof}

Lemma \ref{lemma:bristle-edges} and Lemma \ref{lem:bristle-algo} immediately implies the similarity of
ADFS1 and ADFS2 as follows.
\subsubsection*{Equivalence of ADFS1 and ADFS2}
On insertion of a cross edge, ADFS performs a path reversal and collects the back edges that are now converted to cross edges,
to be iteratively inserted back into the graph.
ADFS2 differs from ADFS1 only by imposing a restriction on the order in which these collected edges are processed.
%ADFS2 is derived from ADFS1 by imposing a restriction on the order in which collected cross edges are processed.
However, for sparse graphs ($m=O(n)$) this restriction does not change its worst case performance 
(see Table \ref{tab:exist_algo}).
Now, Lemma \ref{lem:bristle-algo} states that the time taken by ADFS to incrementally process any number of edges 
is of the order of the time taken to process a sparse graph (with only $2n\log n$ edges). 
Thus, ADFS1 performs similar to ADFS2 even for dense graphs.
%is equivalent to the time taken to process any dense graph. Hence, even for dense graphs ADFS1 performs similar to ADFS2.
%Notice that the worst case bounds of ADFS1 and ADFS2 remain same for sparse graphs (see Figure \ref{fig:exist_algo}).
%Lemma \ref{lem:bristle-algo} states that time taken by the algorithm to process the sparse graph ($2n\log n$ edges) 
%is equivalent to the time taken to process any dense graph. Hence, even for dense graphs ADFS1 performs similar to ADFS2.
Particularly, the time taken by ADFS1 for insertion of any $m\leq {n \choose 2}$ edges is 
$O(n^2\sqrt{\log n})$, i.e., $O(n^{3/2}m_0^{1/2})$ for $m_0=2n\log n$. Thus, we have the following theorem. 
%is artificial because experimentally we have
%seen they perform equally.
%Here the $O(\sqrt{\log n})$ factor is artificial because experimentally we have
%seen they perform equally.

\begin{theorem}
Given a uniformly random sequence of arbitrary length, the expected time complexity of ADFS1 for maintaining a DFS tree 
incrementally is $O(n^2\sqrt{\log n})$.
\label{theorem:theoretical-adfs1}
\end{theorem}

\noindent
\textbf{Remark: }
The factor of $O(\sqrt{\log n})$ in the bounds of ADFS1 and ADFS2 comes from the limitations of our analysis whereas 
empirically their performance matches exactly.

%% file: exp_directed_dag.tex
%\subsection*{Extension to Directed Graphs/DAGs}
The proposed algorithm SDFS2 also works for directed graphs.
It is easy to show that Corollary \ref{cor:bristle-size} also holds for directed graphs 
(with different constants).
Thus, the properties of broomstick structure and hence the analysis of SDFS2 can also be proved for directed graphs 
using similar arguments.
The significance of this algorithm is highlighted by the fact that there \textit{does not} exists any algorithm 
for maintenance of incremental DFS for general directed graphs requiring $o(m^2)$ time. 
Moreover, FDFS also performs very well and satisfies the properties $P_1$ and $P_2$ 
(similar to ADFS in undirected graphs). Note that extension of FDFS for directed graphs is not known to have complexity any better than $O(m^2)$, yet for random directed graphs we can prove it to be $\tilde{O}(n^2)$ using Lemma \ref{lem:bristle-algo}.

\renewcommand{\figW}{0.35\linewidth}
\begin{figure}[h]
	\centering
	\includegraphics[width=\figW]{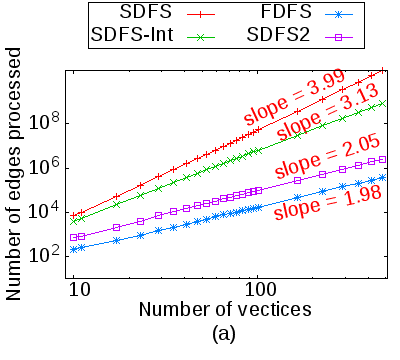}\qquad
	\includegraphics[width=\figW]{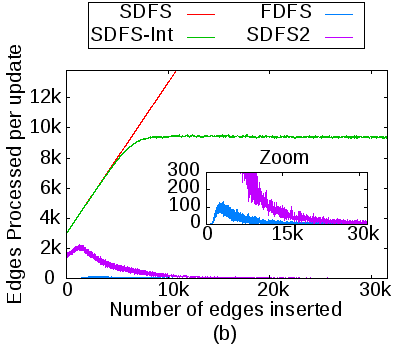}
	\caption{Comparison of existing and proposed algorithms on directed graphs:
		(a) Total number of edges processed for insertion of $m={n \choose 2}$ edges 
		for different values of $n$ in logarithmic scale 
		%(fitting curve to $m+n^{\alpha}$).
		(b)	Number of edges processed per edge insertion 
		for $n=1000$ and up to $n\sqrt{n}$ edge insertions. 	
		%(a) $m=2n$, (b) $m=n\log n$, (c) $m=n\sqrt{n}$, (d) $m={n \choose 2}$.
		See Figure \ref{fig:incDFS-propDT} for corresponding time plot.
	}
	\label{fig:incDFS-propD}
\end{figure}

We now compare the performance of the proposed algorithm SDFS2 with the existing algorithms in the directed graphs. 
Figure \ref{fig:incDFS-propD} (a) compares the total number of edges processed for insertion of $m={n \choose 2}$ edges, 
as a function of number of vertices in the logarithmic scale. 
%Since the number of edges $m$ overshadows the true bounds of SDFS4 and ADFS, the slope of the line is evaluated
%to after removing $m$ edges representing the function $m+O(n^\alpha)$.
%This presents the performance of ADFS and SDFS4 in the right perspective. 
%As expected SDFS2 processes $m+\Omega(n^2)$ edges whereas both SDFS4 and ADFS processes $m+o(n^2)$ edges.
As expected SDFS2 processes $\tilde{O}(n^2)$ edges similar to FDFS. %whereas both SDFS4 and ADFS processes $m+o(n^2)$ edges.
%The predicted complexity of SDFS2, SDFS4 and ADFS matches the performance measured experimentally.
%Moreover, SDFS4 performs even better than ADFS for very high values of $n$. 
Figure \ref{fig:incDFS-propD} (b) compares the number of edges processed 
per edge insertion as a function of number of inserted edges for directed graphs.
Thus, the proposed SDFS2 performs much better than SDFS, and asymptotically equal to FDFS.
%but still lacks behind SDFS4 by a huge margin.
%The same is true for SDFS4, which performs much better than SDFS2 but still lacks behind ADFS by a huge margin.
%between SDFS2 and SDFS4 and equally huge difference between SDFS4 and ADFS around the insertion of $n/2$ edges.
Again despite the huge difference in number of edges processed by SDFS2 with respect to FDFS, 
it is equivalent to FDFS in practice (see Figure \ref{fig:incDFS-propD} (a) and Figure \ref{fig:incDFS-propDT} (a)).
%owing to its sheer simplicity.

\renewcommand{\figW}{0.33\linewidth}
\renewcommand{\figWN}{0.32\linewidth}
\begin{figure}[!h]
	\centering
	\includegraphics[width=\figW]{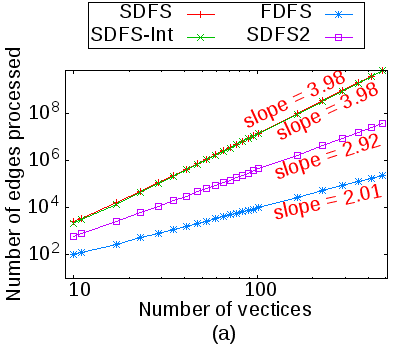}
	\includegraphics[width=\figW]{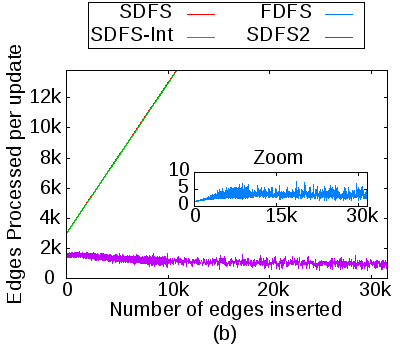}
	\includegraphics[width=\figWN]{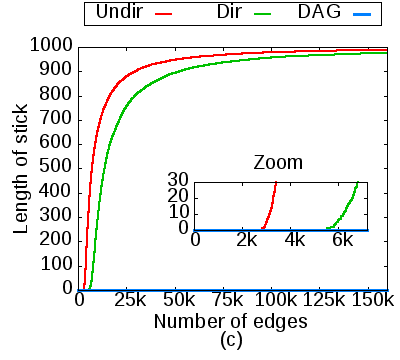}
	\caption{Comparison of existing and proposed algorithms on DAGs:
		(a) Total number of edges processed for insertion of $m={n \choose 2}$ edges 
		for different values of $n$ in logarithmic scale 
		%(fitting curve to $m+n^{\alpha}$).
		(b)	Number of edges processed per edge insertion 
		for $n=1000$ and up to $n\sqrt{n}$ edge insertions. 	
		%(a) $m=2n$, (b) $m=n\log n$, (c) $m=n\sqrt{n}$, (d) $m={n \choose 2}$.
		See Figure \ref{fig:incDFS-propDAGT} for corresponding time plots.
		(c) Comparison of variation of length of broomstick for 1000 vertices and different values of $m$.
		Different lines denote the variation for different type of graphs. Zoomed portion shows the start of each line.
	}
	\label{fig:incDFS-propDAG}
\end{figure}

%\renewcommand{\figW}{0.245\linewidth}
%\begin{figure}[!h]
%	\centering
%	\includegraphics[width=\figW]{figures/IncDFS_D_VarN_SP4L_SDFS_S.png}
%	\includegraphics[width=\figW]{figures/IncDFS_D_VarM_1000_2_S.png}
%	\includegraphics[width=\figW]{figures/IncDFS_DAG_VarN_SP4L_SDFS_S.png}
%	\includegraphics[width=\figW]{figures/IncDFS_DAG_VarM_1000_2_S.png}
%	%\includegraphics[height=33cm]{figures/plot_vary_static_logFit.png}
%	\caption{Comparison of existing and proposed algorithms on directed graphs:
%		(a) Total number of edges processed for insertion of $m={n \choose 2}$ edges 
%		for different values of $n$ in logarithmic scale 
%		%(fitting curve to $m+n^{\alpha}$).
%		(b)	Number of edges processed per edge insertion 
%		for $n=1000$ and up to $n\sqrt{n}$ edge insertions. 	
%		%(a) $m=2n$, (b) $m=n\log n$, (c) $m=n\sqrt{n}$, (d) $m={n \choose 2}$.
%		See Figure \ref{fig:incDFS-propDAGT} for corresponding time plot.
%	}
%	\label{fig:incDFS-propDAG}
%\end{figure}

Finally, we compare the performance of the proposed algorithm SDFS2 with the existing algorithms in DAGs. 
Figure \ref{fig:incDFS-propDAG} (a) compares the total number of edges processed for insertion of $m={n \choose 2}$ edges, 
as a function of number of vertices in the logarithmic scale. 
Both SDFS and SDFS-Int perform equally which was not the case when the experiment was performed on undirected 
(Figure \ref{fig:incDFSvarNM}) or directed graphs (Figure \ref{fig:incDFS-propD}).
Moreover, SDFS2 processes around $\tilde{O}(n^3)$ edges which is more than the proven bound of $\tilde{O}(n^2)$ 
for undirected and directed graphs.
However, FDFS processes $\tilde{O}(n^2)$ edges as expected. %whereas both SDFS4 and ADFS processes $m+o(n^2)$ edges.
%The predicted complexity of SDFS2, SDFS4 and ADFS matches the performance measured experimentally.
%Moreover, SDFS4 performs even better than ADFS for very high values of $n$. 
Figure \ref{fig:incDFS-propDAG} (b) compares the number of edges processed per edge insertion as a function of number 
of inserted edges.
Again, both SDFS and SDFS-Int perform similarly and SDFS2 does not perform asymptotically equal to FDFS even for dense graphs.
Notice that the number of edges processed by SDFS2 does not reach a peak and then 
asymptotically move to zero as in case of undirected and general directed graphs.
Also, FDFS performs much better (similar to ADFS for undirected graphs) for DAGs as compared to directed graphs.
Again, despite superior performance on random DAGs, for general DAGs it can be shown that the analysis of FDFS 
is indeed tight (see Appendix~\ref{appn:wc_FDFS}).

%
%\renewcommand{\figW}{0.5\linewidth}
%\begin{figure}[!h]
%	\centering
%	\includegraphics[width=\figW]{figures/broomComp1.png}
%	\caption{Comparison of variation of length of broomstick for 1000 vertices and different values of $m$.
%		Different lines denote the variation for different type of graphs. Zoomed portion shows the start of each line.
%	}
%	\label{fig:broomC}
%\end{figure}

To understand the reason behind this poor performance of SDFS-Int and SDFS2 on DAGs,
we compare the variation in length of broomstick for the undirected graphs, general directed graphs and DAGs in Figure \ref{fig:incDFS-propDAG} (c).
The length of the broomstick varies as expected for undirected and general directed graphs but always remains zero for DAGs. 
This is because the stick will appear only if the first neighbor of the pseudo root $s$ visited 
by the algorithm is the first vertex (say $v_1$) in the topological ordering of the graph. 
Otherwise $v_1$ hangs as a separate child of $s$ because it not reachable from any other vertex in the graph.
Since the edges in $G(n,m)$ model are permuted randomly, with high probability $v_1$ may not be the first vertex to get connected to $s$.
The same argument can be used to prove branchings at every vertex on the stick. Hence, with high probability there would be some bristles even on the pseudo root $s$. This explains why SDFS-Int performs equal to SDFS as it works same as SDFS 
until all the vertices are visited. SDFS2 only benefits from the inserted edges being reverse cross edges which are valid 
in a DFS tree and hence avoids rebuilding on every edge insertion. Thus, Corollary \ref{cor:bristle-size} and hence the bounds for SDFS2 proved in Theorem \ref{theorem:theoretical-sdfs2} are not valid for the case of DAGs as resulting in performance described above.
Moreover, the absence of the broomstick phenomenon can also be proved for other models of random graphs for DAGs \cite{CordeiroMPTVW10}
using the same arguments.\\

%% file: semi_stream.tex
%\subsection{Semi-streaming algorithms}
\subsection*{Semi-streaming algorithms}
\label{sec:semi-stream}

In the streaming model we have two additional constraints. 
Firstly, the input data can be accessed only sequentially in the form of a stream.
The algorithm can do multiple passes on the stream, but cannot access the entire stream.
Secondly, the working memory is considerably smaller than the size of the entire input stream. 
For graph algorithms, a semi-streaming model allows the size of the working memory to be $\tilde{O}(n)$.

The DFS tree can be trivially computed using $O(n)$ passes over the input graph in the 
semi-streaming environment, each pass adding one vertex to the DFS tree. However, computing 
the DFS tree in even $\tilde{O}(1)$ passes is considered hard~\cite{Farach-ColtonHL15}.
To the best of our knowledge, it remains an open problem to compute the DFS tree using even 
$o(n)$ passes in any relaxed streaming environment~\cite{ConnellC09,Ruhl03}.
Now, some of the direct applications of a DFS tree in undirected graphs are answering connectivity, 
bi-connectivity and 2-edge connectivity queries. All these problems are addressed efficiently 
in the semi-streaming environment using a single pass by the classical work of 
Westbrook and Tarjan \cite{WestbrookT92}. On the other hand, for the applications of a DFS tree 
in directed graphs as strong connectivity, strong lower bounds of space for single-pass 
semi-streaming algorithms have been shown .
%, strong bi-connectiviy and 2-edge connectivity.
Borradaile et al. \cite{BorradaileMM14} showed that any algorithm 
requires a  a working memory of $\Omega(\epsilon m)$ to answer queries of strong connectivity, 
acyclicity or reachability from a  vertex require with probability greater than $(1+\epsilon)/2$. 
%On the other hand, the direct applications of a DFS tree in directed or directed acyclic graphs  
%includes topological ordering and strong connectivity.  %, strong bi-connectiviy and 2-edge connectivity.
%These problems have been shown to have strong lower bounds of space for single-pass semi-streaming algorithms.
%Borradaile et al. \cite{BorradaileMM14} showed that any algorithm requires a  a working memory of 
%$\Omega(\epsilon m)$ to answer queries of strong connectivity, acyclicity or reachability from a  
%vertex require with probability greater than $(1+\epsilon)/2$. This lower bound also extends to 
%finding topological ordering of the graph as noted by Janthong and Valiant \cite{Janthong14}.

We now propose a semi-streaming algorithm for maintaining Incremental DFS for random graphs.
The key idea to limit the storage space required by this algorithm is to just discard those 
edges from the stream whose at least one endpoint is on the stick of the DFS tree. 
As described earlier, this part of DFS tree corresponding to the stick will never be modified 
by the insertion of any edge. 
%Therefore, these edges can be safely deleted without hampering the correctness of the algorithm. 
If both the endpoints of the edge lie in bristles, we update the DFS tree using ADFS/SDFS2. 
Lemma \ref{lemma:bristle-edges} implies that the expected number of edges stored will be $O(n\log n)$.
In case we use  SDFS2 (for directed graphs) we also delete the non-tree edges incident on the {\em stick}.
Hence, we have the following theorem.
\begin{theorem}
	Given a random graph $G(n,m)$, there exists a single pass semi-streaming algorithm for maintaining the 
	DFS tree incrementally, that requires $O(n\log n)$ space.
	% and total $O(m + n^2\log^2 n)$ time for a stream of edge insertions.
	\label{theorem:semi-streaming}
\end{theorem}

Further, for random graphs even strong connectivity can be solved using a single pass in the 
streaming environment by SDFS2 as follows. 
Now, SDFS2 keeps only the tree edges and the edges in the bristles. 
For answering strong connectivity queries, we additionally store the highest edge from each 
vertex on the stick. The strongly connected components can thus be found by a single traversal 
on the DFS tree \cite{Tarjan72}. 
Thus, our semi-streaming algorithm SDFS2 not only gives a solution for strong connectivity
in the streaming setting but also establishes the difference in its hardness for general graphs 
and random graphs. To the best of our knowledge no such result was known for any graph problem 
in streaming environment prior to our work. Thus, we have the following theorem.

%\textsc{Theorem \ref{theorem:semi-streaming-apps}}
%\textit{
\begin{theorem}
	Given a random graph $G(n,m)$, there exists a single pass semi-streaming algorithm for maintaining 
	a data structure that answers 
	%topological ordering and 
	strong connectivity 
	%, strong bi-connectivity and strong 2-edge connectivity 
	queries in $G$ incrementally, 
	requiring $O(n\log n)$ space.
\end{theorem}

%% file: exp_real_data.tex
\section{Incremental DFS on real graphs}
\label{sec:real_exp}
%\vspace{-.75em}
We now evaluate the performance of existing and proposed algorithms on real graphs.
Recall that for random graphs, bristles represent the entire DFS tree until the insertion of $\Theta(n\log n)$ edges.
This forces SDFS2 to rebuild the whole tree requiring total $\Omega(n^2)$ time even for sparse random graphs, 
whereas ADFS and FDFS only partially rebuild the DFS tree and turn out to be much better for sparse random graphs 
(see Figure \ref{fig:incDFS-prop} (b), \ref{fig:incDFS-propD} (b) and \ref{fig:incDFS-propDAG} (b)). 
%Recall that the time taken by SDFS2 is $\Omega(n^2)$ even for sparse graphs (see proof of Theorem \ref{theorem:theoretical-sdfs2}), 
%which is not the case with ADFS/FDFS (recall Figure \ref{fig:incDFS-prop} (b)). 
Now, most graphs that exist in real world are known to be sparse \cite{Melancon06}.
Here again, both ADFS and FDFS perform much better  %(which are not similar to random graphs) 
as compared to SDFS2 and other existing algorithms.
Thus, we propose another simple variant of SDFS (SDFS3), which is both easy to implement and 
performs very well even on real graphs (much better than SDFS2).

\input{propAlgo}

\subsection{Experimental Setup}
%\vspace{-.75em}
%\subsection{Experimental Setup and Datasets}
%\vspace{-.75em}
The algorithms are implemented in C++ using STL (standard template library),
and built with GNU g++ compiler (version 4.4.7) with optimization flag $-O3$.
%We used GNU g++ compiler (version 4.4.7) with optimization flag $-O3$.
The correctness of our code was exhaustively verified on random inputs by ensuring 
the absence of anti-cross edges (or cross edge) in directed (or undirected) graphs.
%by checking that every non-tree edge is not a cross edge.
%This verification was performed over exhaustively
%over generated random graphs. These verifiers were not used while evaluating the time taken on real graphs.
%Our experiments were run on Intel Xeon X5570 2.93 GHz 2 CPU-Nehelam (8-cores per node) on 
%HP-Proliant-BL-280c-Gen6 servers with 48 GB of RAM per node. 
Our experiments were run on Intel Xeon E5-2670V 2.5 GHz 2 CPU-IvyBridge (20-cores per node) on 
HP-Proliant-SL-230s-Gen8 servers with 1333~MHz DDR3 RAM of size 768~GB per node.
Each experiment was performed using a single dedicated processor. %node of the server. 
%Each experiment was performed runs uninterrupted using a single node of the server. 
%Each test uses a single node of the server which runs uninterrupted till the end of execution. 
%%%To highlight the relative performance of different algorithms, we present the time taken by them %by different algorithms
%%%relative to that of ADFS/FDFS (see Appendix \ref{appn:real_exact} for exact times). 
%To determine the relative performance of different algorithms we present the relative time taken by different algorithms
%(see Appendix \ref{appn:real_exact} for exact times). 
%Since the exact time taken depends heavily on the coding style and code structure,
%we present the relative time taken by different algorithms. 

\subsection{Datasets used for evaluation on Real Graphs}
\label{apn:real_datasets}
We consider the following types of graphs in our experiments:
\begin{itemize}
	\item \textbf{Internet topology:} 
	These datasets represent snapshots of network topology on CAIDA project (\textit{asCaida} \cite{LeskovecKF05,snapDATA}),
	Oregon Route Views Project's Autonomous Systems  (\textit{ass733} \cite{LeskovecKF05,snapDATA}) 
	and Internet autonomous systems (\textit{intTop} \cite{ZhangLMZ05,Konect13}).

	\item \textbf{Collaboration networks:} 
	 These datasets represent the collaboration networks as recorded on arxiv's High-Energy-Physics 
	 groups of Phenomenology (\textit{arxvPh} \cite{LeskovecKF07,bigDND14,Konect13}) and Theory 
	 (\textit{arxvTh} \cite{LeskovecKF07,bigDND14,Konect13}), and on DBLP (\textit{dblp} \cite{DBLP02,bigDND14,Konect13}). 

	\item \textbf{Online communication:}
	These datasets represent communication of linux kernel messages (\textit{lnKMsg} \cite{Konect13}), 
	Gnutella peer-to-peer file sharing network (\textit{gnutella} \cite{RipeanuIF02,snapDATA}),
	Slashdot's message exchange (\textit{slashDt} \cite{GomezKV08,Konect13}), Facebook's wall posts 
	(\textit{fbWall} \cite{ViswanathMCG09,Konect13}), Democratic National Committee's (DNC) email correspondence
	(\textit{dncCoR} \cite{Konect13}), Enron email exchange (\textit{enron} \cite{Klimt04,Konect13}), 
	Digg's reply correspondence (\textit{digg} \cite{Choudhury09,Konect13}) and 
	UC Irvine message exchange (\textit{ucIrv} \cite{Opsahl09,Konect13})
	
	\item \textbf{Friendship networks:} 
	These datasets represent the friendship networks of Flickr (\textit{flickr} \cite{MisloveKGDB08,Konect13}, 
	Digg (\textit{diggNw} \cite{HoggL12,Konect13}), Epinion	(\textit{epinion} \cite{Massa05,Konect13}), 
	Facebook (\textit{fbFrnd} \cite{ViswanathMCG09,Konect13}) and Youtube (\textit{youTb} \cite{Mislove09,Konect13}). 

	\item \textbf{Other interactions:} 
	These datasets represent the other networks as Chess game interactions (\textit{chess} \cite{Konect13}),
	user loans on Prosper (\textit{perLoan} \cite{Konect13}), hyperlink network of Wikipedia 
	(\textit{wikiHy} \cite{Mislove09,Konect13}), voting in elections on Wikipedia 
	(\textit{wikiEl} \cite{LeskovecHK07,Konect13}) and conflict resolution on Wikipedia 
	(\textit{wikiC} \cite{WikiEd09,Konect13}).
\end{itemize}

In some of these datasets there are some rare instances in which edges are deleted (not present in new snapshot). 
Thus, in order to use these datasets for evaluation of incremental algorithms we ignore the deletion of these edges 
(and hence reinsertion of deleted edges).
Moreover, in several datasets the edges are inserted in form of batches (having same insertion time), 
where the number of batches are significantly lesser than the number of inserted edges. 
Almost all the algorithms (except  FDFS and SDFS3) can be tweaked to handle such batch insertions more efficiently,
updating the DFS tree once after insertion of an entire batch, instead of treating every edge insertion individually.

\begin{table}[!ht]
%\begin{figure}[!ht]
	{ \scriptsize
	\centering
	\def\arraystretch{1.5}
	\resizebox{\columnwidth}{!}{%
	\begin{tabular}{|l|r|r|r|r|r|r|r|r|r|}
%	\begin{tabular}{|l|r|r|r|m{\colW}|m{\colW}|m{\colW}|m{\colW}|m{\colW}|m{\colW}|}
		\hline
		Dataset & $n$ & $m|m^*$ & {\tiny$\frac{m}{n}|\frac{m}{m^*}$} & {ADFS1}& {ADFS2}& {SDFS3}&{SDFS2} & {SDFS}& {WDFS}\\
		\hline		
\textit{ass733} & 7.72K&21.47K&2.78&\textbf{1.00}&{1.88}&{34.12}&{639.50}&{1.13K}&{2.99K}\\
&&721.00&29.77&\textbf{1.00}&{2.71}&{38.43}&{35.57}&{54.14}&{95.43}\\
\textit{intTop} & 34.76K&107.72K&3.10&\textbf{1.00}&{2.14}&{111.32}&{3.78K}&{8.15K}&{14.65K}\\
&&18.27K&5.89&\textbf{1.00}&{6.07}&{99.47}&{320.49}&{1.83K}&{2.24K}\\
\textit{fbFrnd} & 63.73K&817.03K&12.82&\textbf{1.00}&{2.18}&{146.58}&{2.02K}&{14.67K}&{11.75K}\\
&&333.92K&2.45&\textbf{1.00}&{8.10}&{141.07}&{491.24}&{7.63K}&{4.27K}\\
\textit{wikiC} & 116.84K&2.03M&17.36&\textbf{1.00}&{1.82}&{249.45}&{3.09K}&$>${22.56K}&$>${22.56K}\\
&&205.59K&9.86&\textbf{1.00}&{2.26}&{246.49}&{2.69K}&{4.39K}&{3.35K}\\
\textit{arxvTh} & 22.91K&2.44M&106.72&\textbf{1.00}&{1.81}&{28.31}&{3.41K}&$>${39.96K}&{9.72K}\\
&&210.00&11.64K&\textbf{1.00}&{6.74}&{32.01}&{8.63}&{13.24}&{2.84K}\\
\textit{arxvPh} & 28.09K&3.15M&112.07&\textbf{1.00}&{2.38}&{57.94}&{2.54K}&$>${36.29K}&{11.32K}\\
&&2.26K&1.39K&\textbf{1.00}&{8.25}&{70.75}&{103.23}&{192.22}&{3.17K}\\
\textit{dblp} & 1.28M&3.32M&2.59&\textbf{1.00}&{1.60}&$>${22.07K}&$>${22.07K}&$>${22.07K}&$>${22.07K}\\
&&1.72M&1.93&\textbf{1.00}&{1.84}&$>${21.26K}&$>${21.26K}&$>${21.26K}&$>${21.26K}\\
\textit{youTb} & 3.22M&9.38M&2.91&\textbf{1.00}&{3.53}&$>${347.00}&$>${347.00}&$>${347.00}&$>${347.00}\\
&&203.00&46.18K&{1.26}&{2.26}&$>${322.18}&{1.00}&\textbf{1.00}&{260.73}\\
\hline	
	\end{tabular}
}}
%	\caption{Comparison of time taken by different algorithms (relative to ADFS1) 
%		for maintaining incremental DFS on real undirected graphs.
%		See Figure \ref{fig:real_dataE} for corresponding table comparing exact time taken by different algorithms.
%		%		 Comparison shows Edges Processed comparing to ADFS.
%	} 
 % \setlength{\belowcaptionskip}{1pt} % Chosen fairly arbitrarily
 % \setlength{\abovecaptionskip}{-20pt} % Chosen fairly arbitrarily
	\caption{Comparison of time taken by different algorithms, relative to the fastest (shown in bold), for maintaining
	incremental DFS on real undirected graphs. 
%	In case the time exceeds $100$hrs, 
%	the ratio shows the relative time of more than ($>$) $100$hrs.
	See Table \ref{tab:real_dataE} for corresponding table comparing the exact performance of different algorithms.\newline
%	See Table \ref{tab:real_dataE} for corresponding table comparing the exact time 
%		and memory used by different algorithms.\newline
		%		 Comparison shows Edges Processed comparing to ADFS.
	} 
	\label{tab:real_data}
%\end{figure}
\end{table}

\begin{table}[!t]
%\begin{figure}[!ht]
	\centering
	{\scriptsize
	\def\arraystretch{1.5}
	\resizebox{.75\columnwidth}{!}{%
	\begin{tabular}{|l|r|r|r|r|r|r|r|}
		\hline
		Dataset & $n$ & $m|m^*$ & {\small $\frac{m}{n}|\frac{m}{m^*}$} & \multicolumn{1}{c|}{FDFS} & \multicolumn{1}{c|}{SDFS3} & \multicolumn{1}{c|}{SDFS2} & 
		\multicolumn{1}{c|}{SDFS} \\
		\hline
		\textit{dncCoR} & 1.89K&5.52K&2.92&{1.55}&\textbf{1.00}&{2.27}&{9.86}\\
&&4.01K&1.38&{1.55}&\textbf{1.00}&{2.00}&{7.18}\\
\textit{ucIrv} & 1.90K&20.30K&10.69&{1.69}&\textbf{1.00}&{2.25}&{21.81}\\
&&20.12K&1.01&{1.78}&\textbf{1.00}&{2.35}&{22.14}\\
\textit{chess} & 7.30K&60.05K&8.22&{1.94}&\textbf{1.00}&{2.54}&{20.00}\\
&&100.00&600.46&{52.04}&{26.14}&\textbf{1.00}&\textbf{1.00}\\
\textit{diggNw} & 30.40K&85.25K&2.80&\textbf{1.00}&{1.33}&{3.60}&{14.50}\\
&&81.77K&1.04&\textbf{1.00}&{1.38}&{3.78}&{11.96}\\
\textit{asCaida} & 31.30K&97.84K&3.13&\textbf{1.00}&{4.31}&{13.60}&{64.71}\\
&&122.00&801.98&{12.57}&{42.62}&{1.01}&\textbf{1.00}\\
\textit{wikiEl} & 7.12K&103.62K&14.55&{1.01}&\textbf{1.00}&{2.58}&{51.80}\\
&&97.98K&1.06&{1.00}&\textbf{1.00}&{2.53}&{52.38}\\
\textit{slashDt} & 51.08K&130.37K&2.55&{1.03}&\textbf{1.00}&{2.78}&{5.85}\\
&&84.33K&1.55&{1.04}&\textbf{1.00}&{2.07}&{3.79}\\
\textit{lnKMsg} & 27.93K&237.13K&8.49&{1.82}&\textbf{1.00}&{2.40}&{23.24}\\
&&217.99K&1.09&{1.77}&\textbf{1.00}&{2.30}&{23.13}\\
\textit{fbWall} & 46.95K&264.00K&5.62&{1.29}&\textbf{1.00}&{2.49}&{14.84}\\
&&263.12K&1.00&{1.31}&\textbf{1.00}&{2.73}&{17.11}\\
\textit{enron} & 87.27K&320.15K&3.67&\textbf{1.00}&{1.55}&{5.66}&{67.58}\\
&&73.87K&4.33&\textbf{1.00}&{1.48}&{2.61}&{14.00}\\
\textit{gnutella} & 62.59K&501.75K&8.02&{1.23}&\textbf{1.00}&{2.54}&{19.13}\\
&&9.00&55.75K&{1.17K}&{1.04K}&{1.03}&\textbf{1.00}\\
\textit{epinion} & 131.83K&840.80K&6.38&{1.32}&\textbf{1.00}&{2.29}&{17.77}\\
&&939.00&895.42&{95.27}&{93.62}&\textbf{1.00}&{1.00}\\
\textit{digg} & 279.63K&1.73M&6.19&\textbf{1.00}&{1.18}&{3.96}&$>${29.28}\\
&&1.64M&1.05&\textbf{1.00}&{1.34}&{4.08}&$>${30.92}\\
\textit{perLoan} & 89.27K&3.33M&37.31&\textbf{1.00}&{7.10}&{30.70}&$>${639.03}\\
&&1.26K&2.65K&{2.13}&{13.18}&\textbf{1.00}&{1.01}\\
\textit{flickr} & 2.30M&33.14M&14.39&-&-&-&-\\
&&134.00&247.31K&$>${476.50}&$>${476.50}&{1.01}&\textbf{1.00}\\
\textit{wikiHy} & 1.87M&39.95M&21.36&-&-&-&-\\
&&2.20K&18.18K&$>${69.26}&$>${69.26}&\textbf{1.00}&{1.13}\\
\hline
		\end{tabular}
	}}
%	\setlength{\belowcaptionskip}{1pt} % Chosen fairly arbitrarily
 % \setlength{\abovecaptionskip}{-20pt} % Chosen fairly arbitrarily
%		\caption{Comparison of time taken by different algorithms (relative to FDFS) 
%		for maintaining incremental DFS on real directed graphs.
%		See Figure \ref{fig:real_data2E} for corresponding table comparing exact time taken by different algorithms.}
	\caption{Comparison of time taken by different algorithms, relative to the fastest (shown in bold),
		for maintaining incremental DFS on real directed graphs.	
%		In case the time exceeds $100$hrs, 
%	the ratio shows the relative time of more than ($>$) $100$hrs.	
	If all algorithms exceed 100hrs giving no fastest algorithm, their corresponding relative time is not shown (-). 
		See Table~\ref{tab:real_data2E} for corresponding table comparing the exact performance of different algorithms.}
%		See Table~\ref{tab:real_data2E} for corresponding table comparing the exact time 
%		and memory used by different algorithms.}
	\label{tab:real_data2}
\end{table}

%\vspace{-.75em}
\subsection{Evaluation}
%\vspace{-.75em}
The comparison of the performance of the existing and the proposed algorithms for real undirected graphs 
and real directed graphs is shown in Table \ref{tab:real_data} and Table \ref{tab:real_data2} respectively.
To highlight the relative performance of different algorithms, we present the time taken by them 
relative to that of the fastest algorithm (see Appendix \ref{appn:real_exact} for the exact time and memory used 
by different algorithms). 
In case the time exceeded $100$hrs %the ratio shows the relative time for more than ($>$) $100$hrs.	
the process was terminated, and we show the relative time 
in the table with a '$>$' sign and the ratio corresponding to $100$hrs. 
%is shown corresponding to $100$hrs with a '$>$' sign. 
%If all algorithms exceed $100hrs$ giving no fastest algorithm, the ratio is not shown $(-)$. 
For each dataset, the first row corresponds to the experiments in which the inserted edges are processed one by one, and the
second row corresponds to the experiments in which the inserted edges are processed in batches 
($m^*$ denotes the corresponding number of batches).
The density of a graph can be judged by comparing the  average degree ($m/n$) with the number of vertices ($n$).
Similarly, the batch density of a graph can be judged by comparing the average size of a batch ($m/m^*$) with 
the number of vertices ($n$).

For undirected graphs, Table \ref{tab:real_data} clearly shows that ADFS1 outperforms 
all the other algorithms irrespective of whether the edges are processed one by one or in batches (except \textit{youTb}). 
Moreover, despite ADFS2 having better worst case bounds than ADFS1, the overhead of 
maintaining its data structure $\cal D$ leads to inferior performance as compared to ADFS1. 
Also, SDFS2 significantly improves over SDFS ($>2$ times). %for real graphs just like in case of random graphs.
%as expected from its behavior in random graphs.
However, by adding a simple heuristic, SDFS3 improves over SDFS2 by a huge margin ($>10$ times)
which becomes even more significant when the graph is very dense (\textit{arxvT} and \textit{arxvPh}). 
Also, note that even SDFS3 performs a lot worse than ADFS ($>30$ times) despite having a profound 
improvement over SDFS2. Further, despite having good worst case bounds, WDFS seems to be only of theoretical 
interest and performs worse than even SDFS in general.
However, if the graph is significantly dense (\textit{fbFrnd}, \textit{wikiC}, \textit{arxvT} and \textit{arxvPh}),
WDFS performs better than SDFS but still far worse than SDFS2.
Now, in case of batch updates, SDFS3 is the only algorithm that is unable to %process the batches as a whole.
exploit the insertion of edges in batches.
Hence, SDFS3 performs worse than SDFS2 and even SDFS if the batch density is significantly high %is significantly dense 
(\textit{arxvTh} and \textit{youTb}). Finally, if the batch density is extremely high %is very dense 
(\textit{youTb}), 
the simplicity of SDFS and SDFS2 results in a much better performance than even ADFS.\\
%\begin{observation}
%For real undirected graphs
%\begin{itemize}
% \setlength\itemsep{-0em}
%\item ADFS1 outperforms all other algorithms unless batch density is extremely high,
%where trivial SDFS is better.
%\item SDFS2 mildly improves SDFS, whereas SDFS3 significantly improves SDFS2.
%\item WDFS performs even worse than SDFS for sparse graphs.% is significantly dense.
%\end{itemize}
%\end{observation}

%\vspace{-.75em}
{\centering
	\fbox{\parbox{\linewidth}{\textbf{Observations:}
For real undirected graphs
\begin{itemize}
\item ADFS outperforms all other algorithms by a huge margin, with ADFS1 always performing
mildly better than ADFS2.
%\item ADFS1 outperforms all other algorithms unless batch density is extremely high,
%where trivial SDFS is better.
\item SDFS2 mildly improves SDFS, whereas SDFS3 significantly improves SDFS2.
\item WDFS performs even worse than SDFS for sparse graphs.% is significantly dense.
\end{itemize}
		}}
}
%
%{\centering
%	\fbox{\parbox{\linewidth}{
%{\bf Observation $1$:}~ ADFS1 outperforms all other algorithms unless batch density is extremely high,
%where trivial SDFS is better.
%		}}
%}

For directed graphs, Table \ref{tab:real_data2} shows that both FDFS and SDFS3 perform almost equally well 
(except \textit{perLoan}) and outperform all other algorithms when the edges are processed one by one. 
In general SDFS3 outperforms FDFS marginally when the graph is dense (except \textit{slashDt} and \textit{perLoan}).
The significance of SDFS3 is further highlighted by the fact that it is much simpler to implement 
as compared to FDFS.  %despite performing equally well.
Again, SDFS2 significantly improves over SDFS ($>2$ times). %just like in case of random graphs.
%as expected from its behavior in random graphs.
Further, by adding a simple heuristic, SDFS3 improves over SDFS2 ($>2$ times), 
and this improvement becomes more pronounced when the graph is very dense (\textit{perLoan}). 
Now, in case of batch updates, both FDFS %and SDFS3 is able to process the batches as a whole.
and SDFS3 are unable to exploit the insertion of edges in batches.
%Now, both FDFS and SDFS3 being unable to process the batches as a whole, 
Hence, they perform worse than SDFS and SDFS2 for batch updates, 
if the average size of a batch is at least $600$. SDFS and SDFS2 perform almost equally well 
in such cases with SDFS performing marginally better than SDFS2 when the batch density is significantly high 
(\textit{asCaida}, \textit{gnutella} and \textit{flickr}).\\

%\begin{observation}
%For real directed graphs
%\begin{itemize}
% \setlength\itemsep{-0em}
%\item FDFS and SDFS3 outperforms all other algorithms unless batch density is high,
%where trivial SDFS is better.
%\item FDFS performs better than SDFS3 for sparse graphs.
%\item SDFS2 mildly improves SDFS, and SDFS3 mildly improves SDFS2.
%\end{itemize}
%\end{observation}

%\vspace{-.75em}
{\centering
	\fbox{\parbox{\linewidth}{\textbf{Observations:}
For real directed graphs
\begin{itemize}
\item FDFS and SDFS3 outperform all other algorithms unless batch density is high,
where trivial SDFS is better.
\item SDFS3 performs better than FDFS in dense graphs.
\item SDFS2 mildly improves SDFS, and SDFS3 mildly improves SDFS2.
\end{itemize}
		}}
}

Overall, we propose the use of ADFS1 and SDFS3 for undirected and directed graphs respectively. 
Although SDFS3 performs very well on real graphs, its worst case time complexity is no better 
than that of SDFS on general graphs (see Appendix \ref{appn:wc_SDFS3}). 
%Note that despite the improvement of SDFS3 over SDFS2, being much better in undirected graphs than 
%directed graphs, ADFS performs much better than SDFS3. This further highlights the significance of 
%ADFS for real graphs. 
%For directed graphs, our proposed algorithm SDFS3 being simpler than FDFS, 
%performs equally well suggesting its use for real graphs. 
Finally, in case the batch density of the input graph is substantially high, we can simply use the trivial SDFS algorithm.\\

\noindent
\textbf{Remark:} 
The improvement of SDFS3 over SDFS2 is substantially better on undirected graphs than on directed graphs.
Even then ADFS1 outperforms SDFS3 by a huge margin. Also, when the batch density is extremely high (\textit{youTb}), 
ADFS1 performs only mildly slower than the fastest algorithm (SDFS). These observations further highlight the 
significance of ADFS1 in practice. 

%% file: propAlgo.tex
\subsection{Proposed algorithms for real graphs (SDFS3)}
\label{appn:real_proposed}
%We have seen that a bristle oriented algorithm works equivalently for sparse graphs and dense graphs.
%However, the bristles represent the whole DFS tree until the insertion of $O(n\log n)$ edges.
%This forces SDFS2 to rebuild the whole tree taking $\Omega(n^2)$ time even for sparse graphs. 
%Whereas, ADFS and FDFS processes very few edges in sparse graphs. 
%
%SDFS4 that satisfies the properties $P_1$, $P_2$ of ADFS/FDFS.
%We give different algorithms for directed and undirected graphs inspired by FDFS and ADFS respectively.
%We propose a bristle oriented algorithm SDFS4 that satisfies the properties $P_1$, $P_2$ of ADFS/FDFS.

The primary reason behind the superior performance of ADFS and FDFS is the partial rebuilding of the DFS tree upon
insertion of an edge. However, the partial rebuilding by SDFS2 is significant only when the broomstick has an appreciable size,
which does not happen until the very end in most of the real graphs. 
With this insight, we propose new algorithms for directed and undirected graphs.
The aim is to rebuild only that region of the DFS tree which is affected by the edge insertion.

%
%
%The primary reason behind the superior performance of ADFS and FDFS is partial rebuilding of the DFS tree upon
%insertion of an edge.
%Even though SDFS2 also performs partial rebuilding, but it only significant when the broomstick has an appreciable size,
%which does not happen until the very end in most real graphs. 
%SDFS4 rebuilds only the region of the DFS tree that is affected by the edge insertion.
%We propose different algorithms for directed and undirected graphs taking insight from FDFS and ADFS respectively as follows.
%%We study the 
%%Our algorithms are simple variants of SDFS but performs only marginally weaker compared to ADFS/FDFS in most cases.

\begin{itemize}

\item \textbf{Undirected Graphs}\\
On insertion of a cross edge $(x,y)$, ADFS rebuilds one of the two {\em candidate} subtrees hanging from $LCA(x,y)$ 
containing $x$ or $y$. We propose algorithm SDFS3 that will rebuild only the smaller subtree (less number of vertices) 
among the two candidate subtrees (say $x\in T_1$ and $y\in T_2$). This heuristic is found to be 
extremely efficient compared to rebuilding one of $T_1$ or $T_2$ arbitrarily.
%Recall that on insertion of any cross edge, ADFS rebuilds one of the two subtrees hanging from 
%$LCA(x,y)$ containing the endpoints of the edge. 
%We propose to always rebuild the smaller subtree (less number of vertices) among the two candidate subtrees described above (say $T_1$ and $T_2$).
%Thus, SDFS4 essentially rebuilds the smaller candidate subtree (say $T_2$) on insertion of a cross edge $e$,
%which then hangs from $T_1$ through the edge $e$. 
The smaller subtree, say $T_2$, can be identified efficiently by simultaneous traversal in both $T_1$ and $T_2$. 
and terminate as soon as either one is completely traversed. This takes time of the order of $|T_2|$.
We then mark the vertices of $T_2$ as unvisited and start the traversal from $y$ in $T_2$,
hanging the newly created subtree from edge $(x,y)$. 
%
%
%
%To determine the smaller subtree, say $T_2$, efficiently, we start a traversal in both $T_1$ and $T_2$ simultaneously, 
%and terminate as soon as either one is completely traversed. This takes time of the order of $|T_2|$.
%We then mark the vertices of $T_2$ as unvisited and start the traversal from the endpoint of $e$ in $T_2$.
%The newly created subtree then hangs from $T_1$ through the edge $e$. 
%%In the absence of theoretical high probability equivalent of Theorem \ref{thm:small-sub},
%%we argue the following empirical bound for SDFS4 (with theoretical bound no better than SDFS2).
%%The expected number of edges processed per edge insertion by SDFS4 is almost the same as the size of the subtree $o(n)$,
%%since the graph is sparse. Thus, for insertion of first $2n\log n$ edges the total time required
%%is $o(n^2)$. 
%
\item  \textbf{Directed Graphs}\\
On insertion of an anti-cross edge $(x,y)$, FDFS rebuilds the vertices reachable from $y$ in the subgraph induced 
by a {\em candidate set} of subtrees described in Section~\ref{sec:exist_algos}.
%the subtrees hanging on the right of $path(LCA(x,y),x)$ until the subtree containing $y$ (say $T'$).
%Recall that for DAGs the entire $T'$ was not rebuilt, rather just its subtrees hanging on the left of $path(LCA(x,y),y)$. 
FDFS identifies this affected subgraph using the DFN number of the vertices.
Thus, this DFN number also needs to be updated separately after rebuilding the DFS tree.
This is done by building an additional data structure while the traversal is performed, which
aids in updating the DFN numbers efficiently.
%FDFS starts a traversal from $(x,y)$
%%This is done by  is performed by starting from $(x,y)$
%identifying the affected subgraph using the DFN number of the vertices. 
%After the traversal, the DFN number of the 
%unvisited vertices of this subgraph are updated using a separate procedure. 
%FDFS starts a traversal from $(x,y)$
%%This is done by  is performed by starting from $(x,y)$
%identifying the affected subgraph using the DFN number of the vertices. 
%After the traversal, the DFN number of the 
%unvisited vertices of this subgraph are updated using a separate procedure. 
%
%This traversal is performed by starting from $(x,y)$
%identifying the affected subgraph using the DFN number of the vertices. After the traversal, the DFN number of the 
%unvisited vertices of this subgraph are updated using a separate procedure. 
We propose SDFS3 to simply mark all the subtrees in this {\em candidate set} as unvisited and proceed the
traversal from $(x,y)$. The traversal then continues from each unvisited root of the subtrees marked earlier,
implicitly restoring the DFN number of each vertex.
%On insertion of any cross edge $(x,y)$, FDFS rebuilds proceeds the traversal from the inserted edge  
%in the subgraph induced by the subtrees hanging from $path(x,LCA(x,y))$ (for DAGs we consider $path(x,y)$).
%After the traversal, the DFN number of the unvisited vertices of this subgraph is updated. 
%In SDFS4, we propose to simply mark all the subtrees on the path as unvisited and proceed the
%traversal along $e$. We then perform the traversal from each root of the subtrees earlier marked unvisited
%implicitly restoring the DFN number of each vertex.

\end{itemize}

%% file: conclusion.tex
\section{Conclusions}
%\vspace{-.5em}
\label{sec:conclusion}
Our experimental study of  existing algorithms for incremental DFS on random graphs
presented some interesting inferences. Upon further investigation, we discovered an 
important property of the structure of DFS tree in random graphs: 
the broomstick structure. We then theoretically proved the variation in length of the 
{\em stick} of the DFS tree as the graph density increases, which also exactly matched the 
experimental results. This led to several interesting applications, including the design 
of an extremely simple algorithm SDFS2. 
This algorithm theoretically matches and experimentally outperforms the 
state-of-the-art algorithm in dense random graphs. 
%This algorithm matches or outperforms the 
%state-of-the-art theoretically as well as experimentally in dense random graphs. 
It can 
also be used as a single pass semi-streaming algorithm for incremental DFS as well as 
strong connectivity in random graphs, which also establishes the difference in hardness 
of strong connectivity in general graphs and random graphs. Finally, for real world graphs, 
which are usually sparse, we propose a new simple algorithm SDFS3 which performs much 
better than SDFS2. Despite being extremely simple, it almost always matches the performance 
of FDFS in directed graphs. However, for undirected graphs ADFS was found to outperform
all algorithms (including SDFS3) by a huge margin motivating its use in practice.

%We evaluated the performance of existing algorithms for maintaining incremental DFS on random graphs 
%and found ADFS and FDFS to perform extremely well as compared to the other algorithms. 
%The key reason for their superior performance was found to be the structure of DFS tree for random graphs.
%We described the broom stick structure of a DFS tree for random graphs and proved tight bounds
%for the variation in length of its stick as the graph becomes denser. This insight not only helped us 
%theoretically prove the empirical performance of ADFS and FDFS, but also led us to design 
%an extremely simple algorithm SDFS2 which matches or outperforms ADFS and FDFS
%theoretically as well as experimentally in dense random graphs.
%Moreover, SDFS2 also works for directed graphs for which there does not exist any nontrivial algorithm 
%with known $o(m^2)$ bound. Furthermore, it can also be used as a single-pass semi-streaming algorithm for 
%computing incremental DFS and strong connectivity for random graphs using $O(n\log n)$ space. This also
%establishes the difference in hardness of strong connectivity in general graphs and random graphs.
%Even for real world graphs, which are usually sparse, we propose a new simple algorithm SDFS3 which 
%performs much better than SDFS2. Despite being extremely simple, it almost always matches the performance 
%of FDFS in directed graphs. 

For future research directions, recall that ADFS (see Inference $I_2$) performs extremely well even on 
sparse random graphs. Similarly, the performance of FDFS and SDFS3 is also very good even on sparse random graphs.
However, none of these have asymptotic bounds any better than $\tilde{O}(n^2)$. 
After preliminary investigation, we believe that the asymptotic bounds for ADFS and FDFS (in DAGs) 
should be $O(m+n\text{polylog} n)$ for random graphs.
Also, we believe that the asymptotic bounds for SDFS3 and FDFS (in directed graphs) 
should be $O(m+n^{4/3}\text{polylog} n)$ for random graphs.
%After preliminary investigation, we believe that the asymptotic bounds for ADFS and FDFS for DAGs and SDFS3 
%for random graphs should be $O(m+n\text{polylog} n)$ and $O(m+n^{4/3}\text{polylog} n)$ respectively.
%Our experimental investigation suggests the asymptotic bounds for ADFS and SDFS4 
%to be $O(m+n\text{polylog} n)$ and $O(m+n^{1/3}\text{polylog} n)$ respectively.
It would be interesting to see if it is possible to prove these bounds theoretically.

%The behavior of ADFS/FDFS was aptly explained for dense random graphs using the broomstick structure.
%However, as noted in Inference $I_2$ (see Section \ref{sec:basic_exp}) ADFS performs very well even on 
%sparse random graphs. Also, the performance of SDFS3 is much superior than SDFS2 even for random graphs,
%but does not have asymptotic bounds any better than that of SDFS2. 
%After preliminary investigation, we believe that the asymptotic bounds for ADFS and SDFS3 
%for random graphs should be $O(m+n\text{polylog} n)$ and $O(m+n^{4/3}\text{polylog} n)$ respectively.
%%Our experimental investigation suggests the asymptotic bounds for ADFS and SDFS4 
%%to be $O(m+n\text{polylog} n)$ and $O(m+n^{1/3}\text{polylog} n)$ respectively.
%It would be interesting to see if it is possible to prove these bounds theoretically.
%%similar high probability
%%bounds for the expected performance of ADFS and SDFS4 in random graphs.
%%The proof lies in the properties of a giant component during the phase transition. 

%% file: appendix.tex
%\appendix

%\newpage
\section{Performance analysis in terms of edges}
\label{appn:perf_edge}

\begin{figure}[h]
\centering
\includegraphics[width=0.55\linewidth]{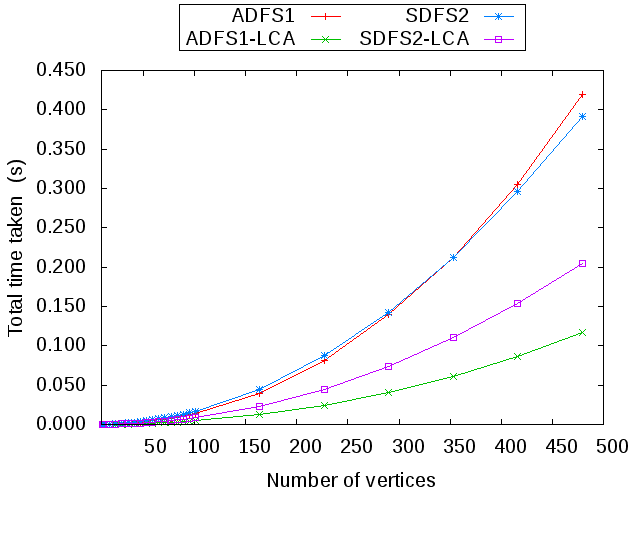}
\caption{Comparison of total time taken and time taken by LCA/LA data structure 
	by the most efficient algorithms for insertion of $m={n \choose 2}$ edges 
	for different values of $n$. 
	%(a) Normal scale.(b) Logarithmic scale.
}
\label{fig:incDFSvarN_TE}
\end{figure}

Most of the algorithms analyzed in this paper require dynamic maintenance of a 
data structure for answering LCA and LA (level ancestor) queries. 
The LCA/LA data structures used by Baswana and Khan \cite{BaswanaK14} takes $O(1)$ amortized time
to maintain the data structure for every vertex whose ancestor is changed in the DFS tree.
However, it is quite difficult to implement and seems to be more of theoretical interest.
Thus, we use a far simpler data structure whose maintenance require $O(\log n)$ time 
for every vertex whose ancestor is changed  in the DFS tree. Figure \ref{fig:incDFSvarN_TE} shows that the time taken by 
these data structures is insignificant in comparison to the total time taken by the algorithm.
Analyzing the number of edges processed instead of time taken allows us to ignore the time taken for maintaining 
and querying this LCA/LA data structure.
Moreover, the performance of ADFS and FDFS is directly proportional to the number of edges processed
along with some vertex updates (updating DFN numbers for FDFS and LCA/LA structure for ADFS).
However, the tasks related to vertex updates can be performed in $\tilde{O}(1)$ time using dynamic trees \cite{Tarjan97}.
Thus, the actual performance of these algorithms is truly depicted by the number of edges processed, justifying
our evaluation of relative performance of different algorithms by comparing the number of edges processed.
\section{Worst Case Input for ADFS1}
\label{appn:wc_ADFS}

\begin{figure*}[!h]
	\centering
	\includegraphics[width=\linewidth]{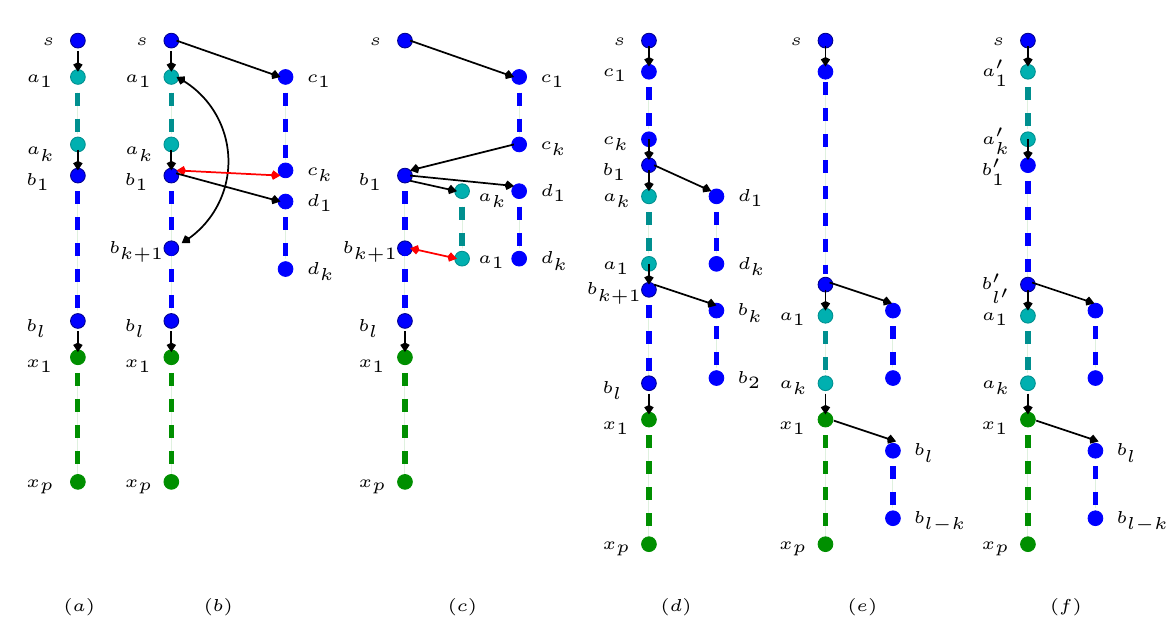}
	\caption{Example to demonstrate the tightness of the ADFS1.
		(a) Beginning of a phase with vertex sets $A$, $B$ and $X$.
		(b) Phase begins with addition of two vertex sets $C$ and $D$. 
		The first stage begins with the insertion of a back edge $(a_1,b_{k+1})$ and a cross edge $(b_1,c_k)$.
		(c) The rerooted subtree with the edges in $A\times X$ and $(b_{k+1},a_1)$  as cross edges.
		(d) Final DFS tree after the first stage.
		(e) Final DFS tree after first phase.
		(f) New vertex sets $A'$, $B'$ and $X$ for the next phase.}
	\label{fig:WorstCase1}
\end{figure*}

We now describe a sequence of $\Theta(m)$ edge insertions for which ADFS1 takes $\Theta(n^{3/2}m^{1/2})$ time. 
Consider a graph $G=(V,E)$ where the set of vertices $V$ is divided into two sets $V'$ and $I$, 
each of size $\Theta(n)$. The vertices in $V'$ are connected in the form of a chain (see Figure \ref{fig:WorstCase1} (a)) 
and the vertices in $I$ are isolated vertices. Thus, it is sufficient to describe only the maintenance 
of DFS tree for the vertices in set $V'$, as the vertices in $I$ will exist as isolated vertices 
connected to the dummy vertex $s$ in the DFS tree (recall that $s$ is the root of the DFS tree).

We divide the sequence of edge insertions into $n_p$ {\em phases}, 
where each phase is further divided into $n_s$ {\em stages}. 
%At the beginning  of each {\em phase}, we extract $2*k$ isolated vertices from $I$ 
%and add them to $V'$, in addition to adding $\Theta(n\sqrt{n})$ edges to the graph.
%In each {\em stage} we add just 2 edges to the graph in such a way that will force ADFS1 
%to process $\Theta(n\sqrt{n})$ edges to rebuild the DFS tree.
%We now describe each phase and its corresponding stages in detail.
At the beginning of each {\em phase}, we identify three vertex sets from the set $V'$, namely 
$A= \{a_1,...,a_k\}$, $B=\{b_1,...,b_l\}$ and $X=\{x_1,...,x_p\}$, 
where $k,l,p\leq n$ are integers whose values will be fixed later. 
%for some values of $k,l,p\leq n$ to be fixed later. 
%where $k=\sqrt{n}$ and $p=\Theta(n)$.
The value of $l$ is $p$ in the first phase and decreases by $k$ in each subsequent phase. 
Figure \ref{fig:WorstCase1} (a) shows the DFS tree of the initial graph. 
We add $p*k$ edges of the set $A\times X$ to the graph. 
Clearly, the DFS tree does not change since all the inserted edges are back edges.
Further, we extract two sets of vertices $C=\{c_1,...,c_k\}$ and $D=\{d_1,...,d_k\}$ 
from $I$ and connect them in the form of a chain as shown in Figure \ref{fig:WorstCase1} (b).

Now, the {\em first stage} of each phase starts with the addition of the back edge $(a_1,b_{k+1})$ followed
by the cross edge $(b_1,c_k)$. As a result ADFS1 will reroot the tree $T(a_1)$ as shown in the 
Figure \ref{fig:WorstCase1} (c) (extra $O(1)$ vertices from $I$ are added to $C$ to ensure the fall of 
$T(a_1)$ instead of $T(c_1)$). This rerooting makes $(b_{k+1},a_1)$ a cross edge. Moreover,
all $p*k$ edges of set $A\times X$ also become cross edges. 
ADFS1 will collect all these edges and process them in $\Omega(p*k)$ time. 
Since the algorithm can process these edges arbitrarily, it can 
choose to process $(b_{k+1},a_1)$ first. It will result in the final DFS tree as shown in Figure \ref{fig:WorstCase1} (d),
converting all the cross edges in $A\times X$ to back edges and bringing an end to the first stage. 
Note the similarity between Figure \ref{fig:WorstCase1} (b) and Figure \ref{fig:WorstCase1} (d):
the set $C$ is replaced by the set $D$, and the set $D$ is replaced by the top $k$ vertices of $B$.
Hence, in the next stage the same rerooting event can be repeated by adding the edges $(a_k,b_{2k+1})$ and 
$(b_{k+1}, d_k)$, and so on for the subsequent stages. Now, in every stage the length of $B$
decreases by $k$ vertices. Hence, the stage can be repeated $n_s=O(l/k)$ times in the phase,
till $A$ reaches next to $X$ as shown in Figure \ref{fig:WorstCase1} (e).
This completes the first phase. 
Now, in the next phase, first $k$ vertices of the new tree forms the set $A'$ followed by the vertices of $B'$ 
leading up to the previous $A$ as shown in Figure \ref{fig:WorstCase1} (f). 
Since the initial size of $I$ is $\Theta(n)$ and the initial size of $B$ is $l<n$, 
this process can continue for $n_p=O(l/k)$ phases. 
Hence, each phase reduces the size of $I$ as well as $B$ by $k$ vertices.
%each reducing the size of $I$ and $B$ by $k$ vertices.

Hence, at the beginning  of each {\em phase}, we extract $2*k$ isolated vertices from $I$ 
and add $p*k$ edges to the graph. In each {\em stage}, we extract $O(1)$ vertices from $I$ and 
add just 2 edges to the graph in such a way that will force ADFS1 
to process $p*k$ edges to update the DFS tree. Thus, the total number of edges added to the graph is $p*k*n_p$ and the total 
time taken by ADFS1 is $p*k*n_p*n_s$, where $n_p=O(l/k)$ and  $n_s=O(l/k)$. 
Substituting $l=n,p=m/n$ and $k=\sqrt{m/n}$ we get the following theorem for any $n\leq m\leq {n\choose 2}$.

\begin{theorem}
For each value of $n\leq m \leq {n\choose 2}$, there exists a sequence of $m$ edge insertions for which ADFS1 requires $\Theta(n^{3/2}m^{1/2})$ time to maintain the DFS tree.
\label{thm:wc1}
\end{theorem}

\noindent \textbf{Remark:} The core of this example is the rerooting event occurring in each 
stage that takes $\Theta(m^{3/2}/n^{3/2})$ time. This event is repeated systematically $n_s*n_p$ times 
to force the algorithm to take $\Theta(n^{3/2}m^{1/2})$ time. 
However, this is possible only because ADFS1 processes the collected cross edges arbitrarily: 
processing the cross edge $(b_k,a_1)$ first amongst all the collected cross edges.
Note that, had the algorithm processed any other cross edge first (as in case of ADFS2), 
we would have reached the end of a phase in a single stage. 
The overall time taken by the algorithm for this example would then be just $\Theta(m)$.

\section{Worst Case Input for FDFS}
\label{appn:wc_FDFS}
 
We now describe a sequence of $O(m)$ edge insertions in a directed acyclic graph for which FDFS takes $\Theta(mn)$ time to 
maintain DFS tree.   
Consider a directed acyclic graph $G=(V,E)$ where the set of vertices $V$ is divided into two sets $A=\{a_1,a_2,...,a_{n/2}\}$ 
and $B=\{b_1,b_2,...,b_{n/2}\}$, each of size $n/2$. The vertices in both $A$ and $B$ are connected in the form of a chain 
(see Figure \ref{fig:WorstCase2} (a), which is the DFS tree of the graph). 
Additionally, set of vertices in $B$ are connected using $m-n/2$ edges (avoiding cycles), i.e.
there can exist edges between $b_i$ and $b_j$, where $i<j$. For any $n\leq m\leq {n\choose 2}$, we can add $\Theta(m)$ edges to $B$ as described above.
Now, we add $n/2$ more edges as described below.

 \begin{figure}[!h]
\centering
\includegraphics[width=0.5\linewidth]{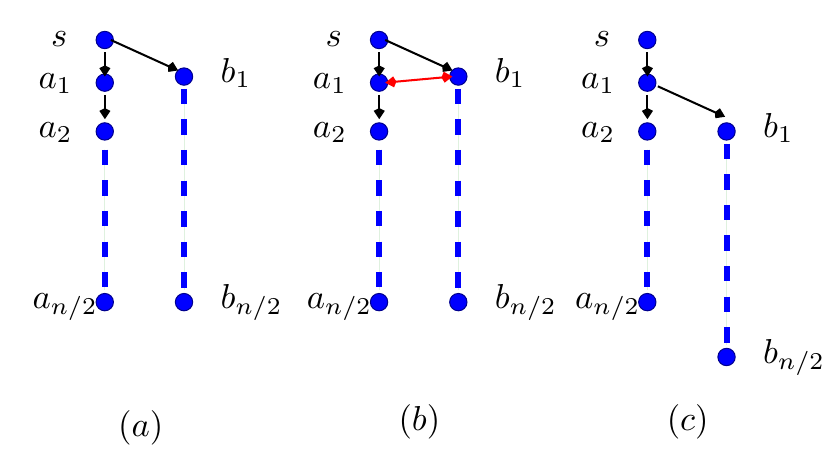}
\caption{Example to demonstrate the tightness of analysis of FDFS.
(a) Initial DFS tree of the graph $G$.
(b) Insertion of a cross edge $(a_1,b_1)$. 
(c) The resultant DFS tree.
}
\label{fig:WorstCase2}
\end{figure}

We first add the edge $(a_1,b_1)$ as shown in Figure \ref{fig:WorstCase2} (b). On addition of an edge $(x,y)$,
FDFS processes all outgoing edges of the vertices having rank $\phi(x)<\phi'\leq \phi(y)$, where $\phi$ is the 
post order numbering of the DFS tree. Clearly, the set of such vertices is the set $B$. Hence, all the 
$\Theta(m)$ edges in $B$ will be processed to form the final DFS tree as shown in Figure \ref{fig:WorstCase1} (c). 
We next add the edge $(a_2,b_1)$ which will again lead to processing of all edges in $B$, and so on. 
This process can be repeated $n/2$ times adding each $(a_i,b_1)$, for $i=1,2,...,n/2$ iteratively. 
Thus, for $n/2$ edge insertions, FDFS processes $\Theta(m)$ edges each, requiring a total of $\Theta(mn)$ time to maintain the DFS tree. 
Hence, overall time required for insertion of $m$ edges is $\Theta(mn)$, as FDFS has a worst case bound of $O(mn)$. 
Thus, we have the following theorem.

\begin{theorem}
For each value of $n\leq m\leq {n\choose 2}$, there exists a sequence of $m$ edge insertions for which FDFS requires $\Theta(mn)$ time to maintain the DFS tree.
\label{thm:wc2}
\end{theorem}

 %\newpage
\section{Worst Case Input for SDFS3}
\label{appn:wc_SDFS3}
 
We now describe a sequence of $m$ edge insertions for which SDFS3 takes $\Theta(m^2)$ time. 
Consider a graph $G=(V,E)$ where the set of vertices $V$ is divided into two sets $V'$ and $I$, 
each of size $\Theta(n)$. The vertices in $V'$ are connected in the form of a three chains (see Figure \ref{fig:WorstCase3} (a)) 
and the vertices in $I$ are isolated vertices. Thus, it is sufficient to describe only the maintenance 
of DFS tree for the vertices in set $V'$, as the vertices in $I$ will exist as isolated vertices 
connected to the dummy vertex $s$ in the DFS tree (recall that $s$ is the root of the DFS tree).

We divide the sequence of edge insertions into $k$ {\em phases}, 
where each phase is further divided into $k$ {\em stages}. 
%At the beginning  of each {\em phase}, we extract $2*k$ isolated vertices from $I$ 
%and add them to $V'$, in addition to adding $\Theta(n\sqrt{n})$ edges to the graph.
%In each {\em stage} we add just 2 edges to the graph in such a way that will force ADFS1 
%to process $\Theta(n\sqrt{n})$ edges to rebuild the DFS tree.
%We now describe each phase and its corresponding stages in detail.
At the beginning of each {\em phase}, we identify three chains having vertex sets from the set $V'$, namely 
$A= \{a_1,...,a_k\}$, $X=\{x_1,...,x_p\}$ in the first chain, $B=\{b_1,...,b_l\}$ and $Y=\{y_1,...,y_q\}$ in the second chain
and $C= \{c_1,...,c_k\}$, $Z=\{z_1,...,z_r\}$ in the third chain as shown in Figure \ref{fig:WorstCase3} (a).
The constants $k,p,q,r=\Theta(\sqrt{m})$ such that $q>r+k$ and $p \approx q+r+k$.
We then add $e_Z=\Theta(m)$ edges to the set $Z$, $e_y=e_z+k+1$ edges to $Y$ 
and $e_x = e_z+e_y$ edges to $X$, which is overall still $\Theta(m)$.
The size of $A$ and $C$ is $k$ in the first phase and decreases by $1$ in each the subsequent phases. 
Figure \ref{fig:WorstCase3} (a) shows the DFS tree of the initial graph. 

Now, the {\em first stage} of the phase starts with addition of the cross edge $(b_1,c_1)$ as shown in Figure \ref{fig:WorstCase3} (b). 
Clearly, $s$ is the LCA of the inserted edge and SDFS3 would rebuild the smaller of the two subtrees $T(b_1)$ and $T(c_1)$.
Since $q>r$, SDFS3  would hang $T(c_1)$  through edge $(b_1,c_1)$ and perform partial DFS on $T(c_1)$ requiring to process $\Theta(m)$ edges in $Z$.
This completes the first stage with the resultant DFS tree shown in the Figure \ref{fig:WorstCase3} (c).
This process continues for $k$ stages, where in $i^{th}$ stage, $T(c_1)$ would initially hang from $b_{i-1}$ and $(b_i,c_1)$ 
would be inserted. The DFS tree at the end of $k^{th}$ stage is shown in Figure \ref{fig:WorstCase3} (d).
At the end of $k$ stages, every vertex in $B$ is connected to the vertex $c_1$, hence we remove it from $C$ for the next phase.
For this we first add the edge $(a_1,c_1)$. Since both $T(b_1)$ and $T(a_1)$ have approximately same number of vertices
(as $p \approx q+r+k$), we add constant number of vertices (if required) to $Z$ from $I$ to ensure $T(b_1)$ is rebuilt.
The resultant DFS tree is shown in Figure \ref{fig:WorstCase3} (e).
Finally, we add $(a_2,c_1)$. Again both $T(c_1)$ and $T(a_2)$ have approximately same number of vertices, 
so we add constant number of vertices from $I$ to $X$ ensuring $T(a_2)$ is rebuild as shown in Figure \ref{fig:WorstCase3} (f).
Note the similarity between Figures \ref{fig:WorstCase3} (a) and \ref{fig:WorstCase3} (f).
In the next phase, the only difference is that $A'=\{a_2,...,a_k\}$, $C'=\{c_2,...,c_k\}$ and $s'=c_1$.
In each phase one vertex each from $A$ and $C$ are removed and constant number of vertices from $I$ are removed.
Hence the phase can be repeated $k$ times. 

Thus, we have $k$ phases each having $k$ stages. Further, in each stage we add a single cross edge forcing SDFS3 to process 
$\Theta(m)$ edges to rebuild the DFS tree. Thus, the total number of edges added to the graph is $k*k=\Theta(m)$ and the total 
time taken by ADFS1 is $k*k*\Theta(m)= \Theta(m^2)$. Hence, we get the following theorem for any $n\leq m\leq {n\choose 2}$.

\begin{figure*}[!h]
\centering
\includegraphics[width=.9\linewidth]{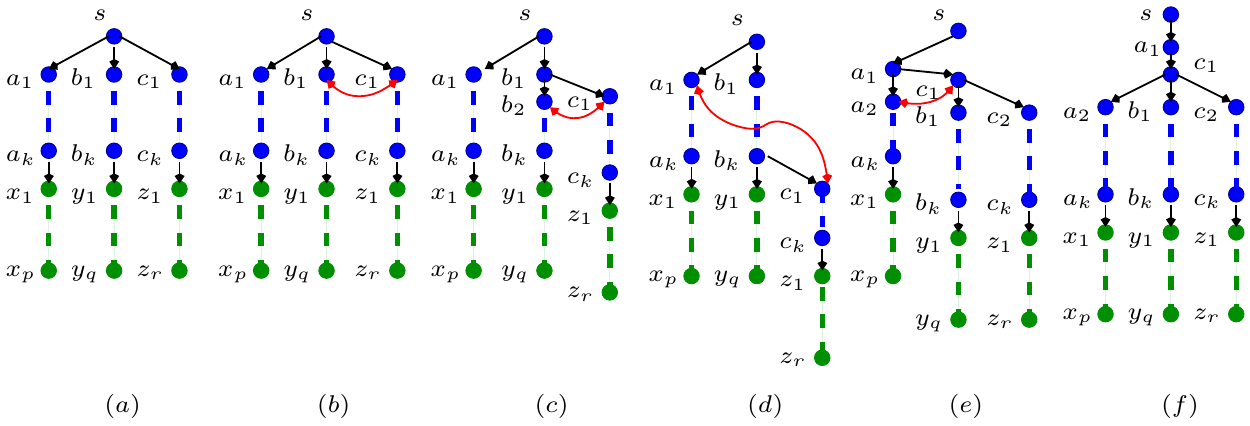}
\caption{Example to demonstrate the tightness of the SDFS3.
(a) Beginning of a phase with vertex sets $A$, $B$ and $X$.
(b) Phase begins with addition of two vertex sets $C$ and $D$. 
The first stage begins by inserting a back edge $(a_1,b_k)$ and a cross edge $(b_1,c_k)$.
(c) The rerooted subtree with the edges in $A\times X$ and $(b_k,a_1)$  as cross edges.
(d) Final DFS tree after the first stage.
(e) Final DFS tree after first phase.
(f) New vertex sets $A'$, $B'$ and $X$ for the next phase.}
\label{fig:WorstCase3}
\end{figure*}

\begin{theorem}
For each value of $n\leq m \leq {n\choose 2}$, there exists a sequence of $m$ edge insertions for which SDFS3 requires $\Theta(m^2)$ time to maintain the DFS tree.
\label{thm:wc3}
\end{theorem}

\noindent \textbf{Remark:} 
The worst case example mentioned above (say $G_1$) would also work without $X,Y$ and $Z$. 
Consider a second example (say $G_2$), where we take size of $A=2*k+2$,$B=k+1$ and $C=k$ and 
the vertices of $C$ have $\Theta(m)$ edges amongst each other. The same sequence of edge insertions
would also force SDFS3 to process $\Theta(m^2)$ edges. However, $G_1$ also ensures the same worst case bound
for SDFS3 if it chooses the subtree with lesser edges instead of the subtree with lesser vertices, 
which is an obvious workaround of the example $G_2$. The number of edges $e_x,e_y$ and $e_z$ are chosen precisely 
to counter that argument.

\section{Time Plots for experiments}
\label{appn:time_plots}
In this section we present the corresponding time plots for experiments performed earlier which 
were measured in terms of number of edges processed.
The comparison of the existing incremental algorithms for random undirected graphs are shown in 
Figure~\ref{fig:incDFSvarNMT}.
The comparison of the existing and proposed algorithms for random undirected graphs, 
random directed graphs and random DAGs are shown in Figure~\ref{fig:incDFS-propT}, 
Figure~\ref{fig:incDFS-propDT} and Figure~\ref{fig:incDFS-propDAGT} respectively.

\renewcommand{\figW}{0.325\linewidth}
\begin{figure*}[!ht]
\centering
\includegraphics[width=\figW]{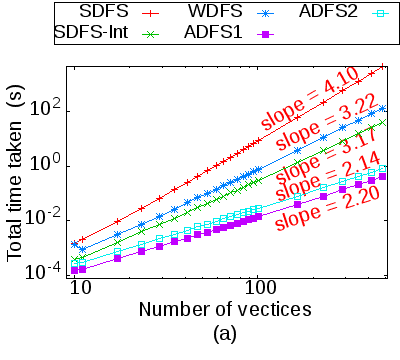}
	\includegraphics[width=\figW]{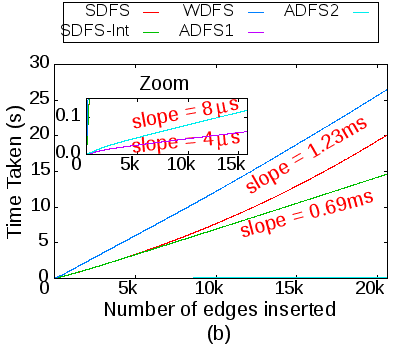}
	\includegraphics[width=\figW]{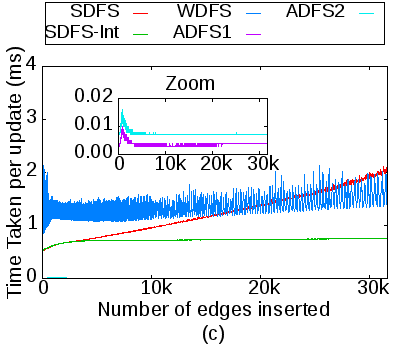}
\caption{
For various existing algorithms, the plot shows 
(a) Total time taken (logarithmic scale) for insertion 
of $m={n \choose 2}$ edges for different values of $n$ ,
(b) Total time taken for $n=1000$ and up to $n\sqrt{n}$ edge insertions, 
(c) Time taken per update for $n=1000$ and up to $n\sqrt{n}$ edge insertions. 
}
\label{fig:incDFSvarNMT}
\end{figure*}

%\renewcommand{\figW}{.43\linewidth}
%\begin{figure*}[h]
%	\centering
%	\includegraphics[width=\figW]{figures/IncDFSN_VarN_SP4_TN.png}
%	\includegraphics[width=\figW]{figures/IncDFSN_VarN_SP4L_TN.png}
%	\caption{Total time taken by existing algorithms for insertion 
%		of $m={n \choose 2}$ edges for different values of $n$. (a) Normal scale.
%		(b) Logarithmic scale.
%	}
%	\label{fig:incDFSvarNT}
%\end{figure*}
%
%\begin{figure*}[!h]
%	\centering
%	\includegraphics[width=\figW]{figures/IncDFS_VarM_1000_SP3_1_TN.png}
%	\includegraphics[width=\figW]{figures/IncDFSN_VarM_1000_SP3_EPI_TN.png}
%	\caption{
%		For $n=1000$ and up to $n\sqrt{n}$ edge insertions the plot shows 
%		(a) Total time taken, (b) Time taken per edge insertion,
%		by the existing algorithms.
%	}
%	\label{fig:incDFSvarMT}
%\end{figure*}

\renewcommand{\figWN}{0.31\linewidth}
\begin{figure*}[!h]
	\centering
	\includegraphics[width=\figW]{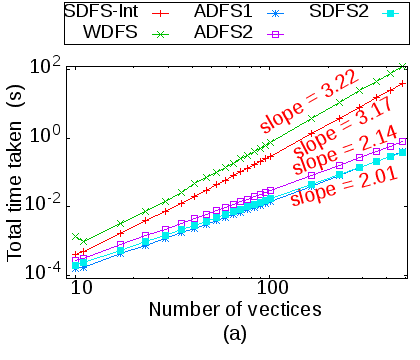}\qquad
	\includegraphics[width=\figWN]{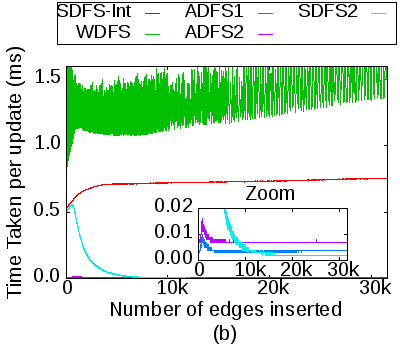}
	\caption{Comparison of existing and proposed algorithms on undirected graphs:
		(a) Total time taken for insertion 
		of $m={n \choose 2}$ edges for different values of $n$.
		% in logarithmic scale 
		%(fitting curve to $m+n^{\alpha}$).
		(b)	Time taken per edge insertion 
		for $n=1000$ and up to $n\sqrt{n}$ edge insertions. 	
		%(a) $m=2n$, (b) $m=n\log n$, (c) $m=n\sqrt{n}$, (d) $m={n \choose 2}$.
	}
	\label{fig:incDFS-propT}
\end{figure*}

\begin{figure*}[!h]
	\centering
	\includegraphics[width=\figW]{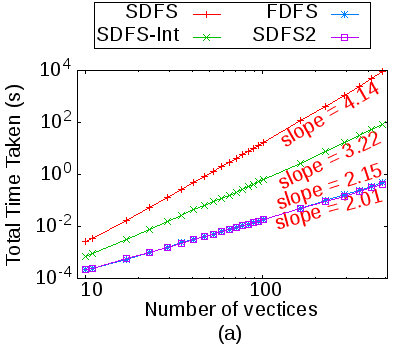}\qquad
	\includegraphics[width=\figW]{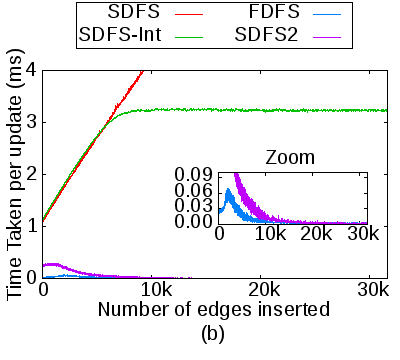}
	\caption{Comparison of existing and proposed algorithms on directed graphs:
		(a) Total time taken for insertion of $m={n \choose 2}$ edges 
		for different values of $n$ in logarithmic scale 
		%(fitting curve to $m+n^{\alpha}$).
		(b)	Time taken  per edge insertion 
		for $n=1000$ and up to $n\sqrt{n}$ edge insertions. 	
		%(a) $m=2n$, (b) $m=n\log n$, (c) $m=n\sqrt{n}$, (d) $m={n \choose 2}$.
	}
	\label{fig:incDFS-propDT}
\end{figure*}

\begin{figure*}[!h]
	\centering
	\includegraphics[width=\figW]{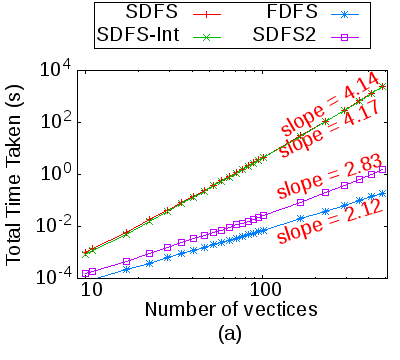}\qquad
	\includegraphics[width=\figW]{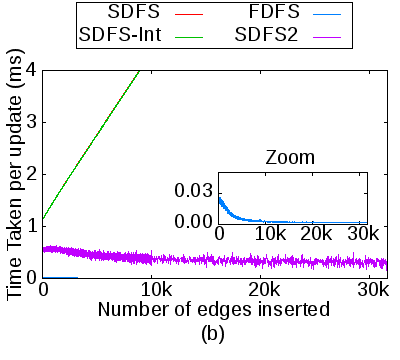}
	\caption{Comparison of existing and proposed algorithms on DAGs:
		(a) Total time taken for insertion of $m={n \choose 2}$ edges 
		for different values of $n$ in logarithmic scale 
		%(fitting curve to $m+n^{\alpha}$).
		(b)	Time taken  per edge insertion 
		for $n=1000$ and up to $n\sqrt{n}$ edge insertions. 	
		%(a) $m=2n$, (b) $m=n\log n$, (c) $m=n\sqrt{n}$, (d) $m={n \choose 2}$.
	}
	\label{fig:incDFS-propDAGT}
\end{figure*}

\newpage
\section{Exact performance comparison for real graphs}
\label{appn:real_exact}
The performance of different algorithms in terms of time and memory required on 
real undirected graphs and real directed graphs is shown in Table~\ref{tab:real_dataE} 
and Table~\ref{tab:real_data2E} respectively.

\setlength{\tabcolsep}{.3em}

\begin{table*}[!ht]
%\begin{figure}[!ht]
	
	\centering
	\def\arraystretch{1}
	\resizebox{\columnwidth}{!}{%
	\begin{tabular}{|l|r|r|r|r|r|r|r|r|r|r|r|r|r|r|r|}
%	\begin{tabular}{|l|r|r|r|m{\colW}|m{\colW}|m{\colW}|m{\colW}|m{\colW}|m{\colW}|}
		\hline
		Dataset & $n$ &$m|m^*$ & { \tiny $\frac{m}{n}|\frac{m}{m^*}$} & \multicolumn{2}{c|}{ADFS1}& \multicolumn{2}{c|}{ADFS2}& \multicolumn{2}{c|}{SDFS3}&
		\multicolumn{2}{c|}{SDFS2} & \multicolumn{2}{c|}{SDFS}& \multicolumn{2}{c|}{WDFS}\\
\hline
\textit{ass733} & 7.72K&21.47K&2.78&\textbf{0.08s}& {35.09M}&{0.15s}& {37.17M}&{2.73s}& {39.91M}&{51.16s}&\textbf{33.25M}&{1.51m}&\textbf{33.25M}&{3.99m}& {450.06M}\\
&&721.00&29.77&\textbf{0.07s}& {35.11M}&{0.19s}& {36.66M}&{2.69s}&\textbf{31.98M}&{2.49s}& {32.77M}&{3.79s}& {32.77M}&{6.68s}& {385.61M}\\
\textit{intTop} & 34.76K&107.72K&3.10&\textbf{0.50s}& {160.14M}&{1.07s}& {162.58M}&{55.66s}& {150.17M}&{31.50m}& {152.12M}&{1.13h}&\textbf{148.95M}&{2.04h}& {2.57G}\\
&&18.27K&5.89&\textbf{0.55s}& {160.80M}&{3.34s}& {169.44M}&{54.71s}& {150.84M}&{2.94m}& {151.81M}&{16.76m}&\textbf{146.00M}&{20.51m}& {2.56G}\\
\textit{fbFrnd} & 63.73K&817.03K&12.82&\textbf{10.26s}& {817.67M}&{22.39s}& {837.23M}&{25.06m}&\textbf{725.08M}&{5.75h}& {747.55M}&{41.80h}& {748.66M}&{33.48h}& {15.13G}\\
&&333.92K&2.45&\textbf{10.46s}& {816.72M}&{1.41m}& {844.47M}&{24.59m}&\textbf{724.12M}&{1.43h}& {745.77M}&{22.17h}& {746.42M}&{12.40h}& {15.14G}\\
\textit{wikiC} & 116.84K&2.03M&17.36&\textbf{15.96s}& {1.89G}&{29.01s}& {1.91G}&{1.11h}&\textbf{1.65G}&{13.71h}& {1.67G}&$>${100.00h}&-&$>${100.00h}&-\\
&&205.59K&9.86&\textbf{15.78s}& {1.89G}&{35.67s}& {1.91G}&{1.08h}&\textbf{1.65G}&{11.77h}& {1.67G}&{19.23h}& {1.67G}&{14.68h}& {38.50G}\\
\textit{arxvTh} & 22.91K&2.44M&106.72&\textbf{9.01s}& {2.10G}&{16.28s}& {2.11G}&{4.25m}& {1.81G}&{8.53h}&\textbf{1.81G}&$>${100.00h}&-&{24.33h}& {39.48G}\\
&&210.00&11.64K&\textbf{8.04s}& {2.61G}&{54.22s}& {2.77G}&{4.29m}& {2.34G}&{1.16m}&\textbf{2.33G}&{1.77m}& {2.33G}&{6.34h}& {3.94G}\\
\textit{arxvPh} & 28.09K&3.15M&112.07&\textbf{9.92s}& {2.69G}&{23.59s}& {2.71G}&{9.58m}& {2.33G}&{7.00h}&\textbf{2.32G}&$>${100.00h}&-&{31.19h}& {50.80G}\\
&&2.26K&1.39K&\textbf{8.15s}& {2.69G}&{1.12m}& {2.70G}&{9.61m}& {2.33G}&{14.02m}&\textbf{2.32G}&{26.11m}&\textbf{2.32G}&{7.18h}& {26.12G}\\
\textit{dblp} & 1.28M&3.32M&2.59&\textbf{16.31s}&\textbf{4.51G}&{26.12s}& {5.14G}&$>${100.00h}&-&$>${100.00h}&-&$>${100.00h}&-&$>${100.00h}&-\\
&&1.72M&1.93&\textbf{16.93s}&\textbf{4.51G}&{31.17s}& {4.76G}&$>${100.00h}&-&$>${100.00h}&-&$>${100.00h}&-&$>${100.00h}&-\\
\textit{youTb} & 3.22M&9.38M&2.91&\textbf{17.29m}&\textbf{13.29G}&{1.02h}& {14.04G}&$>${100.00h}&-&$>${100.00h}&-&$>${100.00h}&-&$>${100.00h}&-\\
&&203.00&46.18K&{23.44m}& {13.27G}&{42.08m}&\textbf{11.55G}&$>${100.00h}&-&{18.67m}& {12.28G}&\textbf{18.62m}& {12.28G}&{80.93h}& {165.80G}\\
\hline
	\end{tabular}
}
	\caption{Comparison of time taken by different algorithms in seconds(s)/minutes(m)/hours(h)
		and memory required in kilobytes(K)/megabytes(M)/gigabytes(G)
		for maintaining incremental DFS on real undirected graphs.
		%		 Comparison shows Edges Processed comparing to ADFS.
	} 
	\label{tab:real_dataE}
%\end{figure}
\end{table*}

\begin{table*}[!ht]
%\begin{figure}[!ht]
	\centering
	
	\def\arraystretch{1}
	\setlength{\tabcolsep}{.8em}
	\resizebox{\columnwidth}{!}{%
	\begin{tabular}{|l|r|r|r|r|r|r|r|r|r|r|r|}
		\hline
		Dataset & $n$ & $m|m^*$ & {\tiny $\frac{m}{n}|\frac{m}{m^*}$} & \multicolumn{2}{c|}{FDFS} & \multicolumn{2}{c|}{SDFS3} & \multicolumn{2}{c|}{SDFS2} & 
		\multicolumn{2}{c|}{SDFS} \\
		\hline
\textit{dncCoR} & 1.89K&5.52K&2.92&{0.34s}& {9.22M}&\textbf{0.22s}& {8.61M}&{0.50s}&\textbf{8.50M}&{2.17s}&\textbf{8.50M}\\
&&4.01K&1.38&{0.34s}& {9.22M}&\textbf{0.22s}& {8.62M}&{0.44s}& {8.48M}&{1.58s}&\textbf{8.47M}\\
\textit{ucIrv} & 1.90K&20.30K&10.69&{0.88s}& {17.47M}&\textbf{0.52s}& {15.30M}&{1.17s}& {22.94M}&{11.34s}&\textbf{15.17M}\\
&&20.12K&1.01&{0.87s}& {17.47M}&\textbf{0.49s}& {15.28M}&{1.15s}& {15.16M}&{10.85s}&\textbf{15.14M}\\
\textit{chess} & 7.30K&60.05K&8.22&{26.03s}& {44.58M}&\textbf{13.39s}& {38.42M}&{34.00s}&\textbf{38.20M}&{4.46m}& {45.98M}\\
&&100.00&600.46&{26.54s}& {52.48M}&{13.33s}&\textbf{43.75M}&\textbf{0.51s}& {43.94M}&\textbf{0.51s}& {44.64M}\\
\textit{diggNw} & 30.40K&85.25K&2.80&\textbf{2.23m}& {85.05M}&{2.98m}& {82.64M}&{8.03m}& {82.36M}&{32.33m}&\textbf{81.58M}\\
&&81.77K&1.04&\textbf{2.32m}& {85.95M}&{3.21m}&\textbf{77.41M}&{8.77m}& {80.56M}&{27.70m}& {77.92M}\\
\textit{asCaida} & 31.30K&97.84K&3.13&\textbf{35.21s}& {97.45M}&{2.53m}& {87.00M}&{7.98m}& {87.95M}&{37.98m}&\textbf{86.48M}\\
&&122.00&801.98&{35.19s}& {91.64M}&{1.99m}& {86.92M}&{2.82s}&\textbf{80.75M}&\textbf{2.80s}&\textbf{80.75M}\\
\textit{wikiEl} & 7.12K&103.62K&14.55&{16.21s}& {65.69M}&\textbf{16.02s}& {61.53M}&{41.27s}& {66.11M}&{13.83m}&\textbf{56.23M}\\
&&97.98K&1.06&{16.27s}& {69.36M}&\textbf{16.24s}& {63.88M}&{41.16s}& {59.02M}&{14.18m}&\textbf{58.25M}\\
\textit{slashDt} & 51.08K&130.37K&2.55&{14.30m}& {127.23M}&\textbf{13.85m}& {123.47M}&{38.50m}&\textbf{116.55M}&{1.35h}& {122.88M}\\
&&84.33K&1.55&{15.24m}& {126.61M}&\textbf{14.61m}& {122.08M}&{30.21m}&\textbf{116.94M}&{55.39m}& {122.12M}\\
\textit{lnKMsg} & 27.93K&237.13K&8.49&{9.52m}& {161.89M}&\textbf{5.22m}& {139.08M}&{12.51m}& {139.38M}&{2.02h}&\textbf{138.66M}\\
&&217.99K&1.09&{9.51m}& {161.91M}&\textbf{5.38m}& {138.72M}&{12.35m}& {139.58M}&{2.07h}&\textbf{138.66M}\\
\textit{fbWall} & 46.95K&264.00K&5.62&{27.11m}& {200.05M}&\textbf{21.08m}& {175.05M}&{52.52m}&\textbf{174.80M}&{5.21h}& {174.81M}\\
&&263.12K&1.00&{29.63m}& {200.03M}&\textbf{22.68m}& {175.03M}&{1.03h}&\textbf{174.77M}&{6.47h}& {174.80M}\\
\textit{enron} & 87.27K&320.15K&3.67&\textbf{10.32m}& {258.92M}&{16.00m}& {240.08M}&{58.40m}&\textbf{235.11M}&{11.63h}& {235.12M}\\
&&73.87K&4.33&\textbf{11.31m}& {258.94M}&{16.80m}& {240.08M}&{29.48m}&\textbf{234.94M}&{2.64h}&\textbf{234.94M}\\
\textit{gnutella} & 62.59K&501.75K&8.02&{29.11m}& {345.39M}&\textbf{23.64m}& {286.66M}&{1.00h}&\textbf{284.78M}&{7.54h}& {284.88M}\\
&&9.00&55.75K&{13.49m}& {288.19M}&{11.96m}&\textbf{245.27M}&{0.71s}&\textbf{245.27M}&\textbf{0.69s}& {245.28M}\\
\textit{epinion} & 131.83K&840.80K&6.38&{3.28h}& {556.69M}&\textbf{2.50h}& {487.06M}&{5.71h}&\textbf{478.86M}&{44.34h}& {479.22M}\\
&&939.00&895.42&{3.09h}& {640.75M}&{3.04h}& {580.38M}&\textbf{1.95m}& {570.61M}&{1.95m}&\textbf{570.59M}\\
\textit{digg} & 279.63K&1.73M&6.19&\textbf{3.42h}& {1.13G}&{4.02h}& {986.58M}&{13.53h}&\textbf{977.58M}&$>${100.00h}&-\\
&&1.64M&1.05&\textbf{3.23h}& {1.13G}&{4.32h}& {986.58M}&{13.20h}&\textbf{977.55M}&$>${100.00h}&-\\
\textit{perLoan} & 89.27K&3.33M&37.31&\textbf{9.39m}& {1.33G}&{1.11h}& {1.31G}&{4.80h}&\textbf{1.31G}&$>${100.00h}&-\\
&&1.26K&2.65K&{9.18m}& {1.33G}&{56.78m}& {1.31G}&\textbf{4.31m}&\textbf{1.31G}&{4.35m}& {1.31G}\\
\textit{flickr} & 2.30M&33.14M&14.39&$>${100.00h}&-&$>${100.00h}&-&$>${100.00h}&-&$>${100.00h}&-\\
&&134.00&247.31K&$>${100.00h}&-&$>${100.00h}&-&{12.69m}&\textbf{15.00G}&\textbf{12.59m}& {15.00G}\\
\textit{wikiHy} & 1.87M&39.95M&21.36&$>${100.00h}&-&$>${100.00h}&-&$>${100.00h}&-&$>${100.00h}&-\\
&&2.20K&18.18K&$>${100.00h}&-&$>${100.00h}&-&\textbf{1.44h}&\textbf{16.99G}&{1.63h}& {16.99G}\\
\hline
	\end{tabular}
	}
	\caption{Comparison of performance of different algorithms in terms of time in seconds(s)/minutes(m)/hours(h)
		and memory required in kilobytes(K)/megabytes(M)/gigabytes(G)
		for maintaining incremental DFS on real directed graphs.
	} 
	\label{tab:real_data2E}
\end{table*}